\documentclass[onefignum,onetabnum]{siamonline0516}
\usepackage{amsfonts, amssymb}
\usepackage{bbm}

\usepackage{mathtools}
\usepackage{etextools}
\usepackage{ifthen}

\usepackage{mathrsfs}

\usepackage{pgfplotstable}
\usepackage{cleveref}

\usepackage[caption=false]{subfig}
\usepackage{diagbox}
\usepackage{makecell}
\usepackage{paralist}

\newcommand{\var}[1]{{#1}} 

\newcommand{\set}[1]{{#1}} 

\newcommand{\tens}[1]{\mathcal{#1}}

\newcommand{\matr}[1]{\mathbf{#1}}

\newcommand{\vect}[1]{\mathbf{#1}} 

\makeatletter
\newcommand{\rrank}[1]{\mathop{\operator@font rrank}\{#1\}}   
\newcommand{\krank}[1]{\mathop{\operator@font krank}\{#1\}}
\newcommand{\trace}[1]{\mathop{\operator@font trace}\{#1\}}
\newcommand{\Diag}[1]{\mathop{\operator@font Diag}\{#1\}}    
\newcommand{\Span}[1]{\mathop{\operator@font Span}\{#1\}}    
\renewcommand{\vec}{\mathop{\operator@font vec}} 
\newcommand{\codim}[1]{\mathop{\operator@font codim} #1}    
\makeatother

\newcommand{\eqdef}{\stackrel{\mathrm{def}}{=}}

\newcommand{\NN}{\mathbb{N}} 
\newcommand{\KK}{\mathbb{K}} 
\newcommand{\RR}{\mathbb{R}} 
\newcommand{\CC}{\mathbb{C}} 
\newcommand{\PP}{\mathbb{P}} 
\newcommand{\bmx}{\begin{bmatrix}}
\newcommand{\emx}{\end{bmatrix}}
\newcommand{\bsm}{\left(\begin{smallmatrix}}
\newcommand{\esm}{\end{smallmatrix}\right)}

\providecommand{\rank}[1]{\operatorname{rank}(#1)} 
\providecommand{\xrank}[1]{\operatorname{rank}_{X}(#1)} 

\DeclarePairedDelimiter{\ceil}{\lceil}{\rceil}

\DeclarePairedDelimiter{\brackets}{[}{]} 

\def\ambspace{\set{A}}
\def\vars{\vect{u}}
\def\varsk#1{u_{#1}}

\def\rGenSdV{\rho}

\newsiamremark{remark}{Remark}

\newcommand{\newsiammyremark}[2]{
  \theoremstyle{plain}
  \theoremheaderfont{\normalfont\itshape}
  \theorembodyfont{\normalfont}
  \theoremseparator{.}
  \theoremsymbol{}
  \newtheorem{#1}{#2}
}

\newcommand{\xparam}{\mathscr{X}}
\newcommand{\pspace}[2]{\Pi_{#1}^{#2}}

\newsiammyremark{example}{Example}
\newsiammyremark{assumption}{Assumption}
\crefname{assumption}{Assumption}{Assumptions}

\title{Identifiability of an X-rank decomposition of polynomial maps\thanks{This work is supported by the ERC project ``DECODA'' no.320594, in the frame of the European program FP7/2007-2013.}}
\author{Pierre Comon, Yang Qi, and Konstantin Usevich\\{\tt pierre.comon@gipsa-lab.fr, yangqi@galton.uchicago.edu, konstantin.usevich@gipsa-lab.fr}}

\begin{document}

\maketitle
\newcommand{\slugmaster}{%
\slugger{siaga}{xxxx}{xx}{x}{x--x}}

\begin{abstract}
In this paper, we study a polynomial decomposition model that arises in problems of system identification, signal processing and machine learning.
We show that this decomposition is a special case of the  X-rank decomposition --- a powerful novel concept in algebraic geometry that generalizes the tensor CP decomposition.
We prove new results on generic/maximal rank  and on identifiability of a particular polynomial decomposition model.
In the paper, we try to make  results and basic tools accessible for  general audience (assuming no knowledge of algebraic geometry or its prerequisites).
\end{abstract}


\section{Introduction: polynomial decompositions}\label{sec:intro}
\subsection{Notation}
We use boldface letters ($\vect{a}, \vect{b}$, ...) for vectors, and boldface capital letters ($\matr{A}$, $\matr{B}$, ...) for matrices.
Given an $m$-dimensional vector space $A$ over a field $\KK$, fix a basis for $A$, then a vector $\vect{a} \in A$ can be identified with an $m \times 1$ matrix, i.e., $\vect{a} = \bmx a_1 & \cdots & a_m\emx^{\top}$, where $\cdot^{\top}$ denotes the transpose.
Thus, $\vect{a}^{\top} \vect{b}$ stands for the matrix multiplication\footnote{Note that this is not the inner product in the case $\KK=\CC$} 
$\vect{a}^{\top} \vect{b} = a_1 b_1 + \cdots + a_m b_m$. 
By $\pspace{m}{d}$  we denote the space of multivariate polynomials in $m$ variables of total degree $\le d$, and we write an element of $\pspace{m}{d}$ in the form $f(\vars)$, where $\vars= \bmx \varsk{1} & \cdots & \varsk{m} \emx^{\top}$.

Standardly, we use $\times$ for Cartesian product of sets, and a shorthand notation $A^{\times d} = A \times \cdots \times A$.
We use $A\oplus B$  for the direct sum\footnote{\emph{i.e.} the Cartesian product $A\times B$ equipped with the vector space structure} of vector spaces, and $\otimes$ for the tensor product.
By $S^d(V)$ or $S^{d}V$ we denote the space of $d$-th order symmetric tensors on an $m$-dimensional vector space $V$ (\emph{i.e.}, $m\times \cdots \times m$ symmetric tensors). In $S^d V$, $\vect{v}^{d}$ means $\vect{v} \otimes \cdots \otimes \vect{v}$.

\subsection{Model and examples}
Let $\KK$ be $\RR$ or $\CC$.
Consider  a multivariate polynomial map $\vect{f}: \KK^{m} \to \KK^{n}$, \emph{i.e.,} a vector
 $\vect{f}(\vars) = \bmx f_1(\vars) & \cdots & f_{n}(\vars) \emx^{\top} \in (\pspace{m}{d})^{\times n}$ of multivariate polynomials of total degree $\le d$ in $m$ variables, (\emph{i.e.}, each $f_i \in \pspace{m}{d}$).
Without loss of generality, in this paper, we assume that $f_k(\vect{0}) = \vect{0}$ (\emph{i.e.}, the constant part of $\vect{f}$ is zero).

Following \cite{Dreesen.etal14-Decoupling}, we say that  $\vect{f}$ has a \textit{decoupled representation}, if it can be expressed as
\begin{equation}\label{eq:decoupling_additive}
 \vect{f}(\vars) = \vect{w}_1 g_1(\vect{v}^{\top}_1\vars) + \cdots + \vect{w}_r g_r(\vect{v}^{\top}_r \vars), 
\end{equation}
where  $\vect{v}_k \in \KK^{m}$, $\vect{w}_k \in \KK^{n}$, and where $g_k(t) =  c_{1,k}t + \ldots + c_{d,k} t^d$ are univariate polynomials over $\KK$.
The problem is often to find a decoupled representation \eqref{eq:decoupling_additive} with $r$ minimum.

\begin{example}[$d=1$]\label{ex:d_1}\quad
In this case, $\vect{f}$ is a linear map, \emph{i.e.} $\vect{f}(\vars) = \matr{F}\cdot\vars$ with $\matr{F}\in \KK^{n\times m}$.
Without loss of generality  we can assume $g_k(t) = t$, and  \eqref{eq:decoupling_additive} becomes a low-rank factorization\footnote{\cref{ex:d_1} shows that  \eqref{eq:decoupling_additive} can be interpreted as a ``low-rank factorization'' of a nonlinear map.} 
\[
\matr{F} =  \vect{w}_1 \vect{v}^{\top}_1 + \cdots + \vect{w}_r \vect{v}^{\top}_r.
\]

\end{example}

The next special case is one of the key examples in this paper.
\begin{example}[$n=1$]\label{ex:n_1}
In this case $\vect{f}$ is a single polynomial $\vect{f}(\vars) = f(\vars)$, and \eqref{eq:decoupling_additive} becomes
\begin{equation}\label{eq:decoupling_one_output}
{f}(\vars) =  g_1(\vect{v}^{\top}_1\vars) + \cdots +  g_r(\vect{v}^{\top}_r\vars),
\end{equation}
since we can assume that $\vect{w}_k = [1]$.
An example of \eqref{eq:decoupling_one_output} is shown in \cref{fig:polydec}.
\end{example}
The decomposition \eqref{eq:decoupling_one_output}%
\begin{itemize}
\item
is known as sum of  \emph{ridge functions} or plane waves \cite{Logan.Shepp75D-Optimal,Oskolkov02SMJ-Representations} in approximation theory;

\item
corresponds to ridge polynomial  neural networks \cite{Shin.Ghosh96IToNN-Ridge} (RPNs) in machine learning;

\item
appears in blind source separation problems in signal processing \cite{Comon.etal15conf-polynomial}.
\end{itemize}
\begin{figure}[hbtp]
\centering
\subfloat[$f(x,y)= 6xy^2 + 4xy$]{\includegraphics[height=2.8cm]{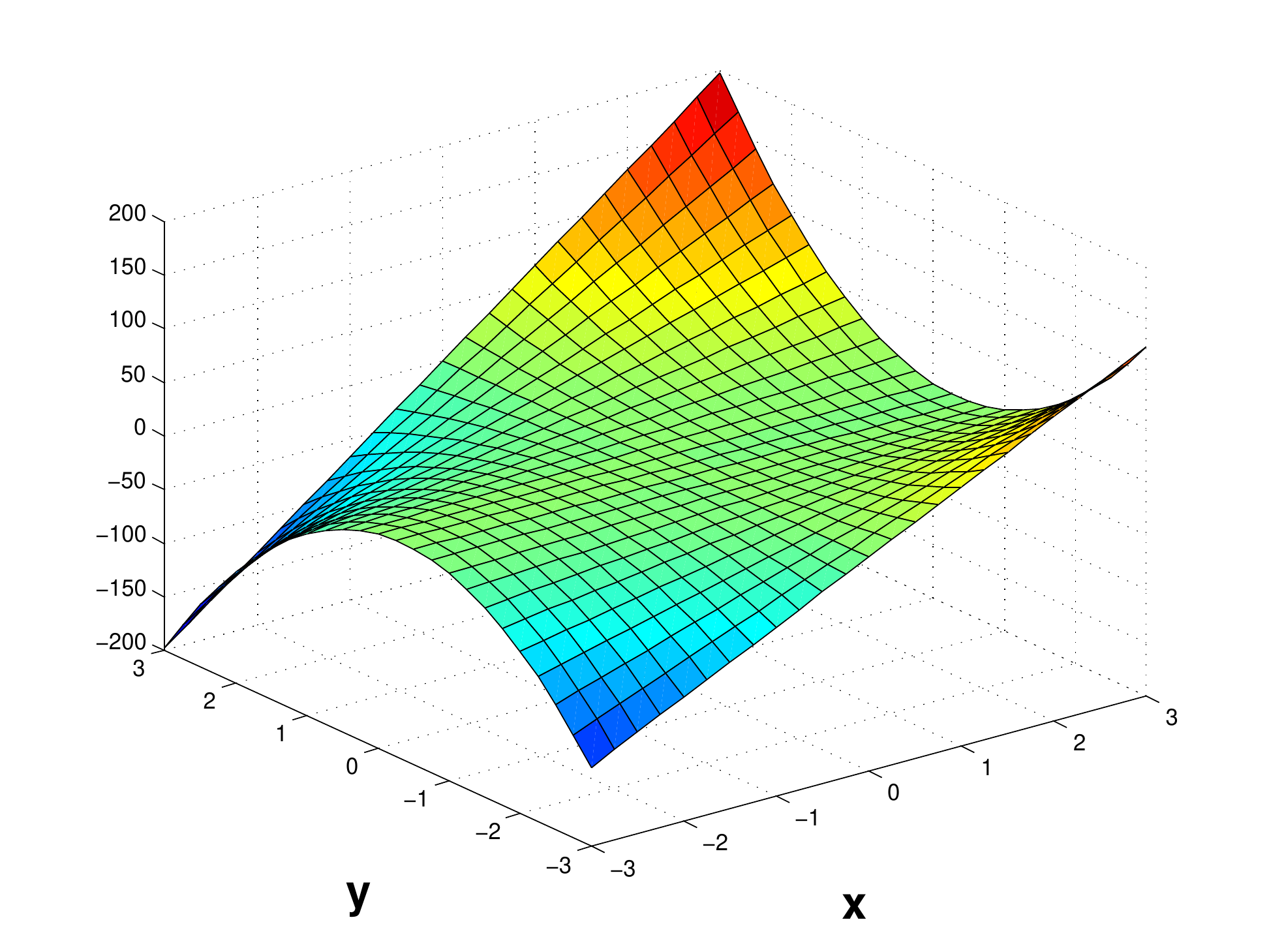}}
\subfloat[$g_1(x+y)$]{\includegraphics[height=2.8cm]{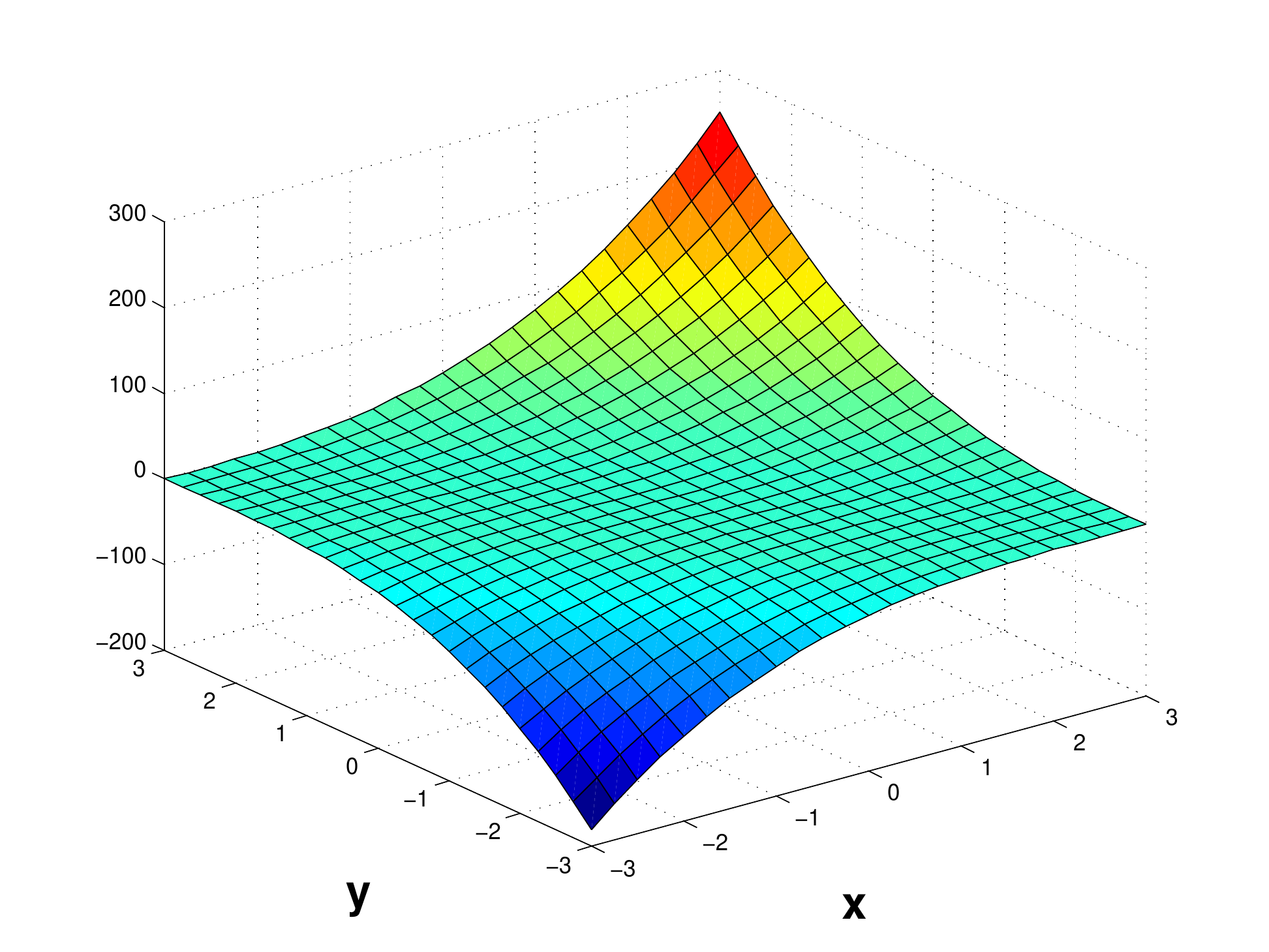}}
\subfloat[$g_2(x-y)$]{\includegraphics[height=2.8cm]{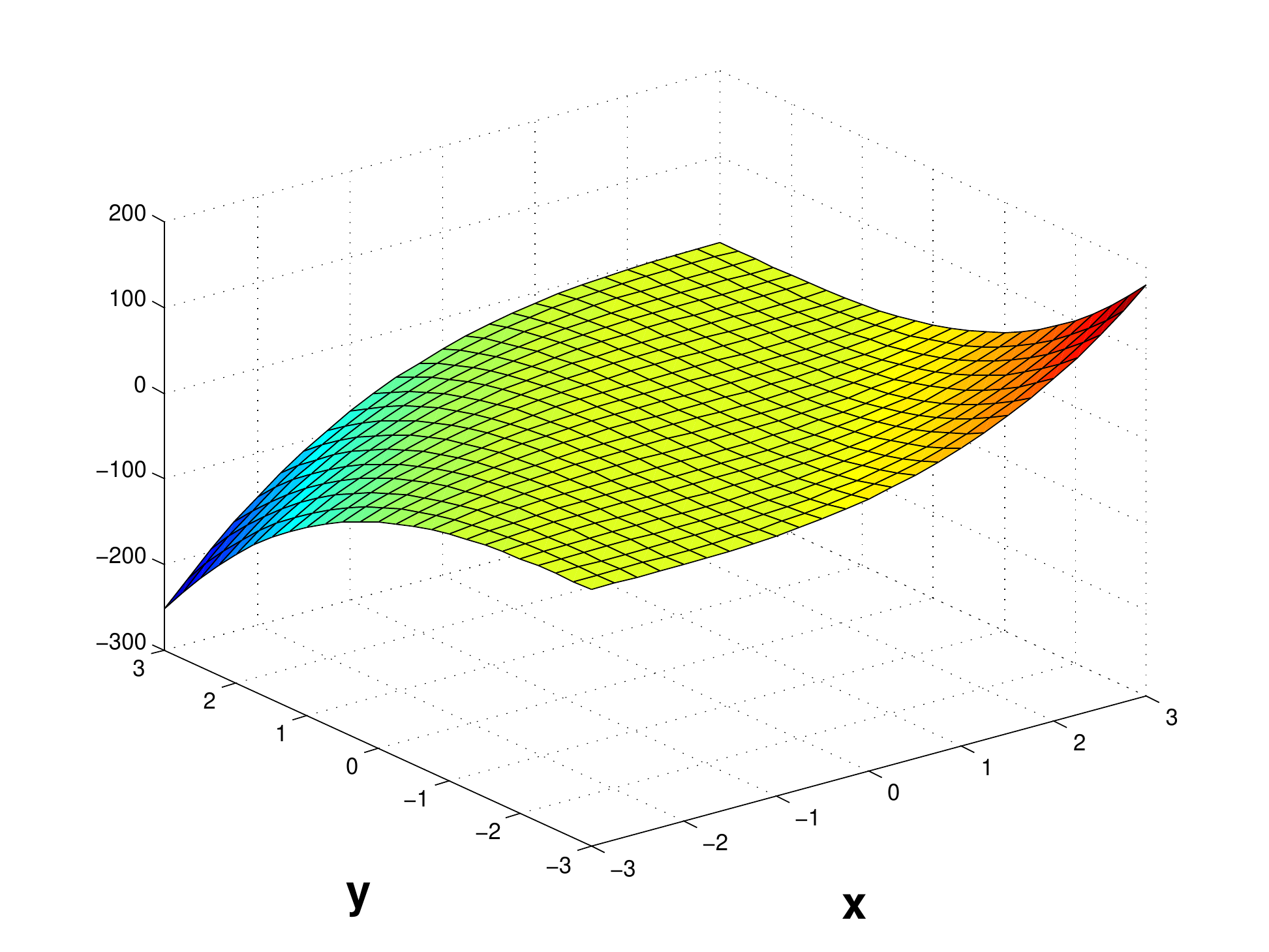}}
\subfloat[$g_3(x)$]{\includegraphics[height=2.8cm]{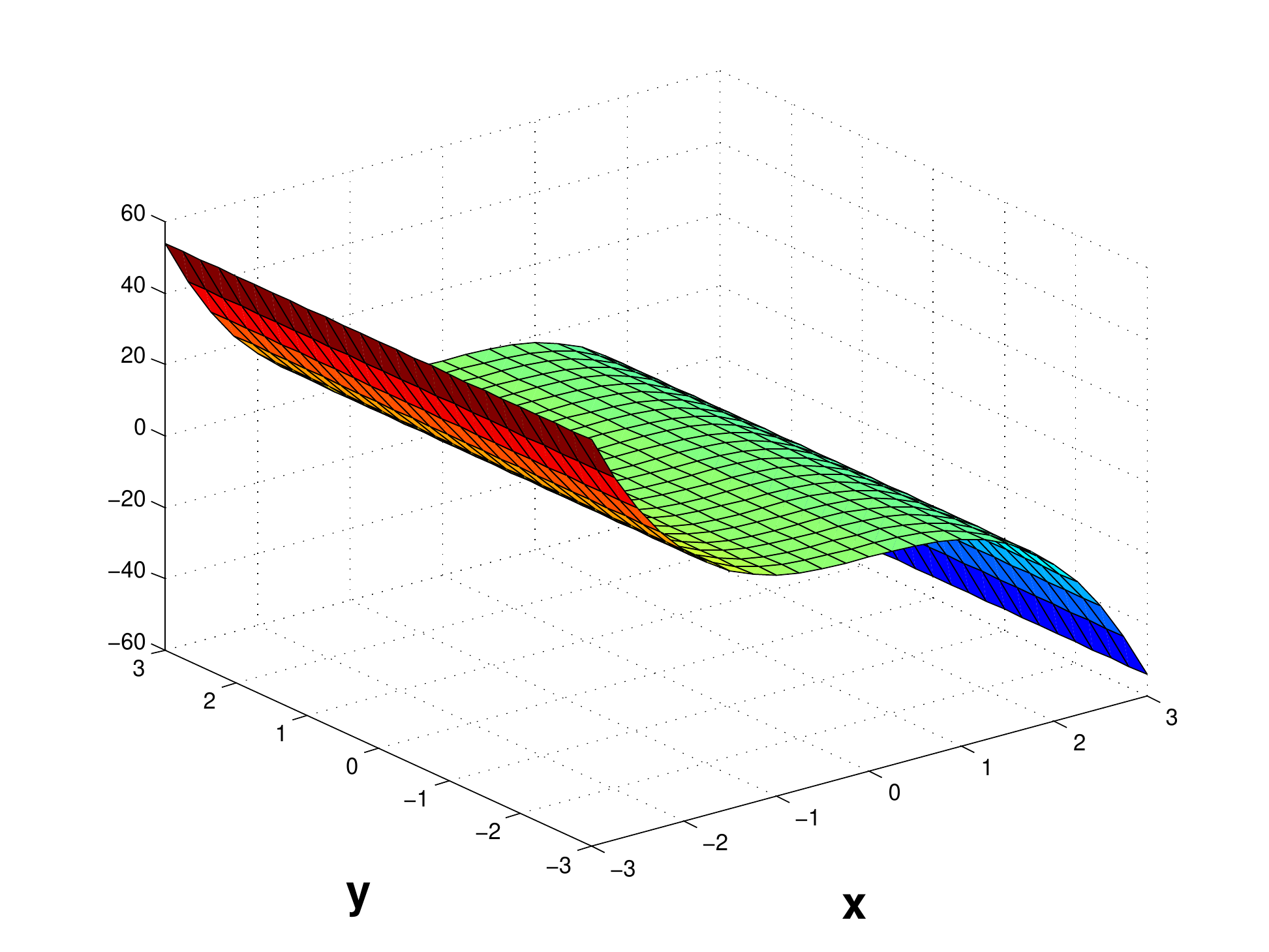}}
\label{fig:polydec}
\caption{$f(x,y) = g_1(x+y) + g_2(x-y) + g_3(x)$, $g_1(t) = t^3+t^2$, $g_2(t) = t^3-t^2$, $g_3(t) = -2t^3$.}
\end{figure}

Next, the homogeneous versions of \cref{eq:decoupling_additive} and \cref{eq:decoupling_one_output} are well-known in algebraic geometry.
\begin{example}[$n=1$, $f$ --- homogeneous]\label{ex:waring}
If $f$ is homogeneous of degree $d$, then $g_k(t)$ should be also homogeneous, \emph{i.e.}  $g_k(t) = c_k t^d$. Hence, the decomposition \eqref{eq:decoupling_additive} becomes
\begin{equation}\label{eq:waring_decomposition}
 f(\vars) =  c_1\cdot (\vect{v}^{\top}_1\vars)^d + \cdots +  c_r \cdot  (\vect{v}^{\top}_r \vars)^d.
\end{equation}
The decomposition \eqref{eq:waring_decomposition} is known as Waring decomposition, and was subject to numerous studies in the literature \cite{Iarobbino.Kanev99-Power,AlexanderHirschowitz95jag-nondefectivity}.
Via the correspondence between homogeneous polynomials  and symmetric tensors (see \cref{sec:symtdec}), \eqref{eq:waring_decomposition} becomes the symmetric tensor decomposition
\begin{equation}\label{eq:symtdec}
f = c_1 \vect{v}_1^d + \cdots + c_r \vect{v}_r^d,
\end{equation}
where $f\in S^dV$ is the symmetric tensor corresponding to the polynomial in $f(\vars)$.
\end{example}
For homogeneous case, the general decomposition (for $n>1$) was also already considered.
\begin{example}[$n>1$, $f$ --- homogeneous]\label{ex:sim_waring}
As in \cref{ex:waring}, \eqref{eq:decoupling_additive} can be rewritten as
\begin{equation}\label{eq:sim_waring_decomposition}
 f(\vars) =  \vect{w}_1\cdot (\vect{v}^{\top}_1\vars)^d + \cdots +  \vect{w}_r \cdot  (\vect{v}^{\top}_r \vars)^d.
\end{equation}
The decomposition \eqref{eq:sim_waring_decomposition} is exactly the simultaneous Waring decomposition of homogeneous polynomials $f_1(\vars), \ldots, f_n(\vars)$ (equivalently, CP decomposition of a partially symmetric tensor). 
\end{example}

\begin{example}[the general case, $n > 1$, $\vect{f}$ --- non-homogeneous]
As summarized in \cite{Dreesen.etal14-Decoupling}, the general decomposition 
\eqref{eq:decoupling_additive} appears in the field of nonlinear system identification \cite{Schoukens.etal14conf-System,Giri.Bai10-Block}.
A common problem in identification (parameter estimation) for several challenging nonlinear block-structured systems (parallel Wiener-Hammerstein \cite{Schoukens.etal14conf-System}  and nonlinear feedback \cite{VanMulders.etal14conf-Identification} models) is to decompose a nonlinear function (represented by a polynomial) in the form \eqref{eq:decoupling_additive}.
\end{example}

\begin{remark}
In the system identification literature (\cite{Dreesen.etal14-Decoupling}), the decomposition \eqref{eq:decoupling_additive} is often written in a compact form
\[
\vect{f}(\vars) = \matr{W} \vect{g}(\matr{V}^{\top} \vars),
\]
where $\matr{V} = \bmx \vect{v}_1 & \cdots & \vect{v}_r \emx \in \KK^{m\times r}$,  $\matr{W} = \bmx \vect{w}_1 & \cdots & \vect{w}_r \emx \in \KK^{n\times r}$ and $\vect{g}: \KK^{r} \to \KK^{r}$ defined as $\vect{g}(t_1,\ldots, t_r) = \bmx g_1(t_1) & \cdots & g_r(t_r)\emx^{\top}$.
Also, a block-diagram for  decomposition \eqref{eq:decoupling_additive} (given in \cref{fig:decoupled}) is often used, where the ``input'' variables $\vect{u}$ are transformed by a linear transformation, followed by component-wise nonlinear transformations. The ``outputs'' are obtained by linear combinations of the results of the nonlinear transformation.
\end{remark}
\begin{figure}[htb!]
    \centering
    \tikz \node [scale=0.8]{
        \begin{tikzpicture}
            \node (u1) at (0,3) {$u_1$};
            \node at (0,2.15) {$\vdots$};
            \node (um) at (0,1) {$u_m$};
            \draw [thick,fill=black!20, rounded corners=5pt] (1,0.5) rectangle (5,3.5); \node (F) at (3,2) {$\vect{f}(u_1,\ldots,u_m)$};
            \draw [->, thick, label=] (5,3) -- (5.65,3) node[right] {$y_1$};
            \node at (5.95,2.15) {$\vdots$};
            \draw [->, thick] (5,1) -- (5.65,1) node[right] {$y_n$};
            \draw [->, thick] (u1) -- (1,3);
            \draw [->, thick] (um) -- (1,1);
        \end{tikzpicture}
        \quad
        \raisebox{6\height}{\Large$=$}
        \quad
        \begin{tikzpicture}
            \node (u1) at (0,3) {$u_1$};
            \node at (0,2.15) {$\vdots$};
            \node (um) at (0,1) {$u_m$};
            \draw [thick] (1,0.5) rectangle (2,3.5); \node (L) at (1.5,2) {$\matr{V}^{\top}$};
            \draw [->, thick] (u1) -- (1,3);
            \draw [->, thick] (um) -- (1,1);
            \node [shape=rectangle,draw,thick,fill=black!20,rounded corners=5pt] (g1) at (4,3) {$g_1(t_1)$};
            \draw [->, thick, label=] (2,3) -- (g1) node[above,midway] {$t_1$};
            \node at (4,2.15) {$\vdots$};
            \node [shape=rectangle,draw,thick,fill=black!20,rounded corners=5pt] (gr) at (4,1) {$g_r(t_r)$};
            \draw [->, thick] (2,1) -- (gr) node[above,midway] {$t_r$};
            \draw [thick] (6,0.5) rectangle (7,3.5); \node (R) at (6.5,2) {$\matr{W}$};
            \draw [->, thick] (g1) -- (6,3) node[above,midway] {$g_1$};
            \draw [->, thick] (gr) -- (6,1) node[above,midway] {$g_r$};
            \node (y1) at (8,3) {$y_1$};
            \node at (8,2.15) {$\vdots$};
            \node (yn) at (8,1) {$y_n$};
            \draw [->, thick] (7,3) -- (y1);
            \draw [->, thick] (7,1) -- (yn);
        \end{tikzpicture} 
    };
    \caption{Representation of a polynomial decomposition.}\label{fig:decoupled}
		 \end{figure}
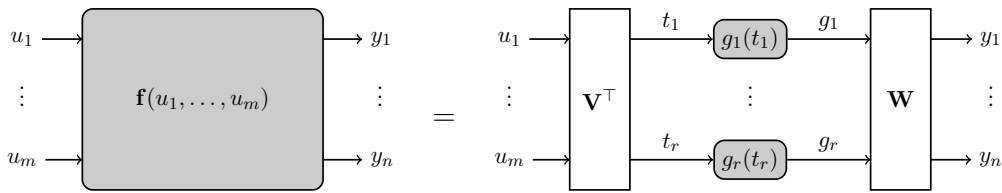

\subsection{Goals and previous works}\label{sec:goals}
When using  model \eqref{eq:decoupling_additive}, a few natural theoretical questions arise that are important to understand the limits of the applicability of the model.

\begin{enumerate}
\item When is the model \emph{identifiable}? (\emph{i.e.}, when is the decomposition \eqref{eq:decoupling_additive}  unique?).
\item What is the upper bound on $r$ in \eqref{eq:decoupling_additive} needed to represent any polynomial?
\item What is the typical (for a ``random'' $\vect{f}$) behavior of $r$ in the shortest decomposition?  
\end{enumerate}

As for the special (homogeneous) cases of decomposition \eqref{eq:decoupling_additive} (Examples 1,3,4), all the three cases were a subject of rapid development in the last two decades, and many results are available.
In this paper, we address the non-homogeneous case (Examples 2 and 5), where very few results are available (listed below).
\paragraph{Bounds on  $r$ and typical behavior} This question was considered only for $n=1$, in the papers 
 \cite{Schinzel02JdTdNdB-decomposition,Schinzel02CM-decomposition,Bialynicki-Birula.Schinzel08CM-Representations}.
The best result shows that  any $f \in \pspace{m}{d}$ can be decomposed as \eqref{eq:decoupling_one_output} whenever
\begin{equation}\label{eq:bound_rmax_BBS}
r \le \binom{m+d-2}{d-1},
\end{equation}
where the bound\footnote{Bound  \eqref{eq:bound_rmax_BBS} is better than a naive bound  $\binom{m+d-1}{d}$ (number of monomials in the highest degree part of $f$).} \eqref{eq:bound_rmax_BBS} is valid for $\RR$, $\CC$ and for certain finite fields.
The typical behavior of $r$ in the shortest decomposition is known only for the case $m=2$ and $n=1$ \cite{Schinzel02JdTdNdB-decomposition} (the case of bivariate polynomials).

\paragraph{Uniqueness} The uniqueness in representations \eqref{eq:decoupling_additive} was almost not studied.
The authors of \cite{Dreesen.etal14-Decoupling} suggested to construct a structured tensor from the coefficients of polynomials.
Based on a Kruskal-type condition for unstructured tensors, they propose a bound for generic uniqueness that depends on $r,m,d$. 
This bound is, however, applicable only to unstructured tensors, and not to the decomposition \eqref{eq:decoupling_additive}, as we argue in \cref{rem:DIS_generic_uniqueness}.

\subsection{Contribution and structure of this paper}
In this paper, we show that that the decomposition \eqref{eq:decoupling_additive} can be viewed as a special case of $X$-rank decomposition.
The notion of $X$-rank (or rank with respect to a variety $\widehat{X}$) is a powerful concept developed in  the field of algebraic geometry that generalizes
 matrix rank, tensor rank, symmetric tensor rank and other notions of rank.
The questions raised in \cref{sec:goals} can be addressed in the framework of X-rank and correspond to finding maximal, typical, generic ranks and to checking $r$-identifiability (generic uniqueness).
In particular, we:
\begin{enumerate}
\item Obtain results on identifiability and partial identifiability of \eqref{eq:decoupling_additive}.
\item Determine the value of generic rank for some special cases of $n=1$.
\item Obtain a new bound on $r_{max}$ (for $\KK =\RR$ or $\CC$) that is better than \eqref{eq:bound_rmax_BBS}.

\end{enumerate}

Although in this paper we do not develop  decomposition algorithms (see \cite{Dreesen.etal14-Decoupling}, \cite{VanMulders.etal14conf-Identification},\cite{Usevich14conf-Decomposing} for  available algorithms),
we believe that the ideas may lead to new or improved algorithms. 

In \cref{sec:xrank}, we introduce the concept of X-rank decompositions and make a review of recent results.
We prefer a very simplistic exposition  and hope that \cref{sec:xrank} may serve as an entry point to the literature on X-rank for a wider audience, including applied mathematicians and engineers.
In \cref{sec:uniqueness}, we recall the definition and known results on generic uniqueness (identifiability), and prove equivalence of different definitions appearing in the literature.
 In \cref{sec:algebraic}, we introduce Veronese scrolls, show that decompositions \eqref{eq:decoupling_additive} and \eqref{eq:decoupling_one_output} are related to $X$-rank decompositions for Veronese scrolls, and give defining equations for this variety.
\Cref{sec:scrolls} contains the main results of the paper, including identifiability of Veronese scrolls and polynomial decompositions, dimensions of secant varieties, and results on generic ranks.

\section{X-rank decompositions}\label{sec:xrank}
The concept of $X$-rank (or rank with respect to a variety) was probably first proposed in \cite{Zak2004}, and popularized in \cite{Blekherman.Teitler14-maximum,Landsberg12-Tensors}.
In this section we give key definitions and basic results, in a simplified form.
In particular, we avoid the use of projective varieties whenever possible.

\subsection{X-rank: definitions}
Consider an $N$-dimensional vector space\footnote{For simplicity, one can think that $\ambspace= \KK^{N}$.} $\ambspace$ over $\KK$, where $\KK$ is $\RR$ or $\CC$.
Assume that a subset $\widehat{X} \subset \ambspace$ is fixed that satisfies the following conditions.
\begin{assumption}\label{as:scale_invariant}
 $\widehat{X}$ is  \emph{scale-invariant}, \emph{i.e.} $\vect{v} \in \widehat{X}$ and $\alpha \in \KK$ implies $\alpha \vect{v}  \in \widehat{X}$.
\end{assumption}%
\begin{assumption}\label{as:non_degenerate}
$\widehat{X}$ is \emph{non-degenerate}, \emph{i.e.} it is not contained in any hyperplane of $\ambspace$.
\end{assumption}
\begin{assumption}\label{as:algebraic_variety}
$\widehat{X}$ is an \emph{algebraic variety}, \emph{i.e.} the zero set of a system of polynomial equations (see also \cref{sec:varieties}).
\end{assumption}

\begin{definition}
Given a subset $\widehat{X} \subset \ambspace$, the \emph{$X$-rank} of any vector $\vect{v}  \in \ambspace$ is defined as  the smallest number of rank-one elements, such that $\vect{v}$ can be represented as their sum: 
\begin{equation}\label{eq:xrank}
\xrank{\vect{v} } = \min r: \vect{v}  = \vect{x}_1 + \cdots +\vect{x}_r, \quad \vect{x}_k \in \widehat{\set{X}}.
\end{equation}
Such a decomposition  with the minimal possible number of terms  is called  the  $X$-rank decomposition. (The rank of $\vect{0}\in \ambspace$, by convention, is zero.)
\end{definition}

\cref{as:scale_invariant} guarantees that the $X$-rank is compatible with linear operations, whereas \cref{as:non_degenerate} ensures that  any vector has an $X$-rank decomposition and that the $X$-rank does not exceed $N$.
The \cref{as:algebraic_variety} allows for an algebraic analysis of $X$-rank decompositions.

The X-rank decomposition can be illustrated in \cref{fig:xrank}. It is also similar in spirit to sparse (atomic) decompositions, that appeared recently in other branches of applied mathematics \cite{Chen.etal01SR-Atomic}.

\begin{figure}[htbp]
  \centering
  \includegraphics[height=1.8cm]{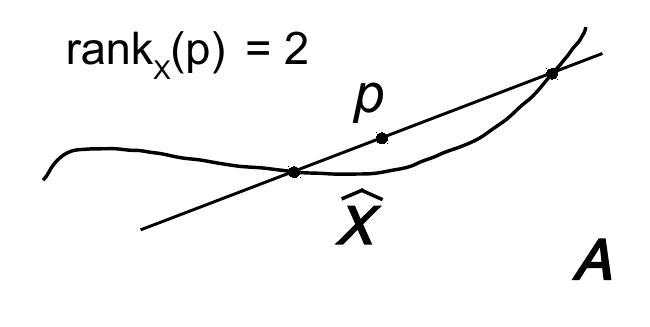}
  \caption{Vector $p$ can be decomposed into the sum of 2 elements of the variety $\widehat{X}$.}
  \label{fig:xrank}
\end{figure}

In fact, \cref{as:scale_invariant,as:algebraic_variety} imply that $\widehat{X}$ is an affine cone of a projective algebraic variety\footnote{where $\PP \ambspace$ is the projective space.} $X \subset \PP \ambspace$.
The projective variety $X$ is the usual starting point in the definition of $X$-rank, see \cite{Zak2004,Blekherman.Teitler14-maximum,Landsberg12-Tensors}.
 In this paper, however, we prefer to work and give definitions in terms of the affine variety $\widehat{X}$, which simplifies some expressions (as we will show later).
One only has to bear in mind that $\dim{X} = \dim{\widehat{X}} -1$.
To avoid pathological phenomena and also for convenience of using algebraic geometry, the following assumption is often imposed.
\begin{assumption}\label{as:irreducible}
$\widehat{X}$ is an irreducible variety (see \cref{sec:varieties}).
\end{assumption}
Finally, for real varieties, the following assumption is often added, to avoid unexpected phenomena and make use of the powerful tools from complex algebraic geometry. 
\begin{assumption}\label{as:smooth_real_point}
The complex variety $\widehat{X}_{\CC}$ is defined by  polynomial equations with real coefficients.
In addition, the corresponding real variety $\widehat{X}_{\RR} = \widehat{X}_{\CC} \cap \RR^{N}$ contains a smooth point of $\widehat{X}_{\CC}$ (see  \cref{sec:varieties}).
\end{assumption}

\subsection{Examples}
The basic examples, considered in \cref{ex:d_1}, \cref{ex:waring} and \cref{ex:sim_waring} fit in the framework of $X$-rank, and are explained in \cref{tab:xrank_examples}. All these examples in \cref{tab:xrank_examples} satisfy \crefrange{as:scale_invariant}{as:smooth_real_point}.

\begin{table}[!hbt]
\caption{Varieties and $X$-ranks}\label{tab:xrank_examples}
\begin{center}
\begin{tabular}{|c|c|c|c|}\hline
Ambient space  ($A$)   & $\dim(A)$ & variety  $\widehat{X}$  & $\dim (\widehat{X})$ \\\hline
$\KK^{I_1} \otimes \cdots \otimes \KK^{I_d}$ &  ${I_1\cdots I_d}$  & $Seg(\KK^{I_1}\times \cdots\times \KK^{I_d}) = \{\vect{a}_1 \otimes \cdots \otimes \vect{a}_d \}$  & $\sum_{k=1}^d I_k - d + 1$ \\
tensor & &  Segre variety & \\\hline
$S^d(\KK^{m})$ &  ${\left(\begin{smallmatrix}m+d-1\\d\end{smallmatrix}\right)}$  & $\nu_d(\KK^{m}) = \{c \vect{a}^{d}\}$ & $m$  \\
symmetric tensor & &  Veronese variety & \\\hline
$ (S^d(\KK^{m}))^{\times n}$ &  ${n \left(\begin{smallmatrix}m+d-1\\d\end{smallmatrix}\right)}$  & $Seg(\KK^{n} \times \nu_d(\KK^{m})) = \{\vect{w} \otimes \vect{a}^{d}\}$ & $m + n-1$  \\
several  & &  Segre-Veronese variety & \\
 symmetric tensors & &   & \\\hline
\end{tabular}
\end{center}
\end{table}

The dimension of the variety of rank-one elements $\widehat{X}$ reflects the number of degrees of freedom in the parameterization of $\widehat{X}$.
Take, for instance, the case of non-symmetric tensors (1-st row in \cref{tab:xrank_examples}). It is 
parameterized by $I_1 + \cdots +I_d$ parameters, but there are $d-1$ redundancies since any element of $\widehat{X}$ has many representations in the form $\vect{a}_1 \otimes \cdots \otimes \vect{a}_d$, due to exchange of scaling.
The other examples in \cref{tab:xrank_examples} follow the same pattern: the dimension of $\widehat{X}$ is equal to the number of parameters minus the number of ``dependencies''.

\subsection{Maximal, typical ranks and  basic relations}
First, we introduce two notations:
\begin{align*}
\Sigma_{\le r,\widehat{\set{X}}} \eqdef \{\mathbf{v} \in A \,|\, \xrank{\vect{v}} \le r \}, \\
\Sigma_{r,\widehat{\set{X}}} \eqdef \{\mathbf{v} \in A \,|\, \xrank{\vect{v}} = r \}. 
\end{align*}
\begin{definition}[Maximal rank]
The maximal $X$-rank is defined as the smallest $r$ such that $\Sigma_{\le r,\widehat{\set{X}}} = \ambspace$, and denoted by $r_{max}$.
\end{definition}
\begin{definition}
A rank $r$ is called typical if  $\Sigma_{ r,\widehat{\set{X}}}$ contains an open Euclidean ball in $\ambspace$.
\end{definition}

Since $\Sigma_{ r,\widehat{\set{X}}}$ is a semialgebraic set \cite{Qi.etal16SJMAA-Semialgebraic}, a rank $r$ is typical if and only if  $\Sigma_{ r,\widehat{\set{X}}}$ has nonzero Lebesgue measure.
Hence, a rank is typical, if and only if it appears with nonzero probability (if the vectors of $\ambspace$ are drawn from an absolutely continuous probability distribution).
The following properties of typical ranks over $\CC$ and $\RR$ are known.
\begin{lemma}\label{lem:generic_rank}
If $\KK = \CC$, there exists only one typical rank, which is called \emph{generic rank}, and denoted by $r_{gen}$.
Moreover, the elements or rank $r_{gen}$ are Zariski-dense in $\ambspace$, \emph{i.e.} there exists an algebraic subvariety $Z \subsetneq \ambspace$ such that $\xrank{\vect{v}} = r_{gen}$ for any $\vect{v} \in A \setminus Z$.
\end{lemma}
\begin{theorem}[\cite{Bernardi.etal16arxiv-real}]
Over the real field, the typical ranks form a contiguous set, \emph{i.e.} there exist the numbers $r_{typ,min}$ and $r_{typ,max}$ such that: 
\begin{itemize}
\item Any $r_1$ such that $r_{typ,min} \le r_1 \le r_{typ,max}$ is typical;
\item Any $r_1$ such that $r_1 < r_{typ,min}$ or $r_1 > r_{typ,max}$ is not typical.
\end{itemize} 
\end{theorem}

Next, the following theorem relates maximal and typical/generic ranks.
\begin{theorem}[ \cite{Blekherman.Teitler14-maximum}]\label{thm:r_gen_r_max}
\begin{itemize}
\item If $\KK = \RR$, then $r_{max} \le 2r_{typ,min}$.
\item If $\KK = \CC$, then $r_{max} \le 2r_{gen}$.
\end{itemize}
\end{theorem}
Finally, there is a relation between real typical ranks and generic complex ranks.

\begin{theorem}[\cite{Blekherman.Teitler14-maximum}]
Let $\widehat{X}_{\RR} = \widehat{X}$ be a real variety satisfying \crefrange{as:scale_invariant}{as:smooth_real_point}, and $\widehat{X}_{\CC} = \widehat{X}_{\RR} \otimes \CC$ be its complexification.
Then it holds that
\[
r_{typ,min} (\widehat{X}_{\RR}) = r_{gen} (\widehat{X}_{\CC}),
\]
\emph{i.e.} the smallest typical real rank is equal to the complex generic rank.
\end{theorem}

All the varieties that we consider in this paper satisfy \crefrange{as:scale_invariant}{as:smooth_real_point}.

\subsection{Secant varieties and border rank}
The $r$-th secant variety\footnote{Here we again prefer using affine varieties. For projective definitions, we invite the reader to consult \cite{Landsberg12-Tensors}.} is, by definition, the Zariski closure of the elements of rank $\le r$:
\[
{\sigma}_r(\widehat{X}) \eqdef \overline{\Sigma_{\le r,\widehat{X}}} \subseteq \ambspace.
\]
The following properties of ${\sigma}_r(\widehat{X})$ are known,  see for example \cite[Section~5.1]{Landsberg12-Tensors} and \cite[Theorem~4.3]{AngeliniBocciChiantini16RICI} for more details. 
\begin{theorem}\label{lem:general_element_secant}
$\phantom{x}$

\begin{itemize}
\item If $\KK = \CC$, then ${\sigma}_r(\widehat{X})$ is the Euclidean closure of $\Sigma_{\le r,\widehat{\set{X}}}$.
\item If $\KK = \CC$, and $\dim {\sigma}_{r-1}(\widehat{X}) < \dim {\sigma}_{r}(\widehat{X})$, then a general point in $\widehat{\sigma}_r(\widehat{X})$ has rank $r$, \emph{i.e.} there exist a subvariety $Y \subsetneq \widehat{\sigma}_r(\widehat{X})$, such that 
\[
{\sigma}_r(\widehat{X}) \setminus Y  \subset \Sigma_{r,\widehat{\set{X}}}.
\]

\item If $\KK = \RR$, it is not the case: there may exist a nonempty Euclidean open subset of ${\sigma}_r(\widehat{X})$ such that each point in this open subset has $X$-rank strictly larger than $r$.
\end{itemize}
\end{theorem}

Nevertheless, there is a correspondence between real and complex varieties \cite{Qi.etal16SJMAA-Semialgebraic}: Let $\widehat{X}_{\RR} = \widehat{X}$ be a real variety satisfying \crefrange{as:scale_invariant}{as:smooth_real_point}, and $\widehat{X}_{\CC} = \widehat{X}_{\RR} \otimes \CC$.
Then for all $r$ the secant variety ${\sigma}_r(\widehat{X}_{\RR})$ satisfies  \crefrange{as:scale_invariant}{as:smooth_real_point},
and ${\sigma}_r(\widehat{X}_{\CC})$ is a complexification of ${\sigma}_r(\widehat{X}_{\RR})$.

\subsection{Defectivity, expected dimension and generic rank}
In this subsection, we only consider the case $\KK = \CC$, and we assume that $\widehat{X}$ satisfies  \crefrange{as:scale_invariant}{as:irreducible}.

A direct consequence of Theorem~\ref{lem:general_element_secant} is that 
the dimensions of $\sigma_{r} (\widehat{X})$ are increasing until $r = r_{gen}$, \emph{i.e.},
\begin{align*}
\dim \widehat{X} = \dim \sigma_{1} (\widehat{X}) & < \dim \sigma_{2} (\widehat{X}) < \cdots  <  \dim \sigma_{r_{gen}-1} (\widehat{X}) \\
&  < \dim \sigma_{r_{gen}} (\widehat{X}) =  \dim\sigma_{r_{gen}+1} (\widehat{X})  = \cdots = \dim A,
\end{align*} 
and tells us that we are able to find the generic rank by looking at dimensions of $\sigma_{r} (\widehat{X})$. For this, a useful concept, i.e., the expected dimension, is introduced.
\begin{definition}[Expected dimension]\label{eq:expdim}
The expected dimension of $\sigma_{r} (\widehat{X})$ is defined as
\[
\exp\dim \sigma_{r} (\widehat{X}) \eqdef \min\{r\dim \widehat{X},\dim A\}
\]
\end{definition}

The intuition behind \cref{eq:expdim} is that if we add in \eqref{eq:xrank} vectors from the variety of dimension $\dim{\widehat{X}}$, we  obtain an object of dimension $r$ times larger. 
In general, \[
\exp\dim \sigma_{r} (\widehat{X}) \ge \dim \sigma_{r} (\widehat{X}).
\]
If there is a strict inequality, $\sigma_{r} (\widehat{X})$ is called defective. Otherwise $\sigma_{r} (\widehat{X})$ is called non-defective.

\begin{corollary}\label{cor:defectivity_r_gen}
The following bound on $r_{gen}$ can be given:
\begin{equation}\label{eq:rank_generic}
r_{gen} \ge \left\lceil\frac{\dim A}{\dim \widehat{X}}\right\rceil
\end{equation}
In particular, if all $\sigma_r(\widehat{X})$ are non-defective, then
$r_{gen} = \left\lceil\frac{\dim A}{\dim \widehat{X}}\right\rceil$.
\end{corollary}

The Alexander-Hirschowitz theorem \cite{AlexanderHirschowitz95jag-nondefectivity} states that for $\widehat{X}= \nu_d(\CC^{m})$, all the secant varieties are non-defective except a finite number of exceptions.
Hence, by Corollary~\ref{cor:defectivity_r_gen} and \cref{tab:xrank_examples}, the generic rank $r_{gen}$ is equal to $\lceil r_{1}(m,d) \rceil$, where
\[
r_1(m,d) \eqdef \frac{\binom{m+d-1}{d}}{m},
\]
except $(m,d) \in\{(3,3), (4,3), (4,5), (4,6)\}$, where $r_{gen}$ is increased by $1$.

\section{Uniqueness and identifiability}\label{sec:uniqueness}
\subsection{Uniqueness of a decomposition}
First, we introduce the notion of uniqueness.
\begin{definition}
An $X$-rank decomposition \eqref{eq:xrank} is \emph{unique} if all the other decompositions  of the form  \eqref{eq:xrank} differ only by permutation of the summands in \eqref{eq:xrank}.
\end{definition}

This definition corresponds to the standard definition of uniqueness of tensor decompositions.
For instance, a tensor decomposition
\begin{equation}\label{eq:tens_decomposition_2}
\tens{T} = \vect{a}_1 \otimes \vect{b}_1 \otimes \vect{c}_1 + \vect{a}_2 \otimes \vect{b}_2 \otimes \vect{c}_2
\end{equation}
is unique if it is unique up to permutation of summands and exchange of scaling in the vectors. 
In this paper, we study the notion of generic uniqueness, or uniqueness of ``almost all'' decompositions.
The following  algebraic definition is often adopted in the literature.
\begin{definition}\label{def:rident_xrank}
A variety $\widehat{\set{X}} \subset \ambspace$ is called $r$-identifiable if a general element in $\Sigma_{r,\widehat{\set{X}}}$ has a unique rank-$r$ decomposition, \emph{i.e.} there exists a semialgebraic subset $Z\subsetneq \Sigma_{r,\widehat{\set{X}}}$ of strictly smaller dimension such that any element in $\Sigma_{r,\widehat{\set{X}}} \setminus Z$  has a unique rank-$r$ decomposition.
\end{definition}

First, we remark on the relation between real and complex identifiability.
\begin{lemma}[\cite{Qi.etal16SJMAA-Semialgebraic}]\label{lem:cplxridenimpliesrealriden}
Assume that $\widehat{X}$ satisfies \crefrange{as:scale_invariant}{as:smooth_real_point},  $r < r_{gen}$ and $\widehat{X}_{\CC}$ is $r$-identifiable. Then $\widehat{X}_{\RR}$ is also $r$-identifiable.
\end{lemma}

Next, we give some interpretation to \cref{def:rident_xrank}.
The following  lemma (\cref{lem:equiv_definition}) states that $\widehat{X}$ is $r$-identifiable if for  ``randomly chosen'' $\vect{p}_1, \dots, \vect{p}_r \in \widehat{X}$  their sum has a  unique $X$-rank decomposition. 
The following proposition  (\cref{cor:param_space}) gives an equivalent definition of identifiability in the parameter space. The proof of both results is given in \cref{sec:basic_generic_equivalence}.
\begin{lemma}\label{lem:equiv_definition}
Let $\KK = \CC$,  $\widehat{X}$ satisfy \crefrange{as:scale_invariant}{as:irreducible}. Then $\widehat{X}$ is $r$-identifiable if and only if
\begin{equation}\label{eq:uniqueness_X}
\text{ for }r\text{ general points }\vect{p}_1, \dots, \vect{p}_r \in \widehat{X}, \vect{p}_1 + \cdots + \vect{p}_r\text{ has a unique rank-}r\text{ decomposition}.
 \end{equation}
\end{lemma}

\begin{proposition}\label{cor:param_space}
Let $\widehat{X}$ be an algebraic variety  over $\KK$ ($\KK =\RR$ or $\CC$)  satisfying \crefrange{as:scale_invariant}{as:smooth_real_point}.
 Assume  that there exists a polynomial map  $\mathscr{X}:  \KK^{M} \to \ambspace$ such that $\widehat{X} = \mathscr{X}(\KK^{M})$. 
Then $\widehat{X}$ is $r$-identifiable if and only if
for a general point $(\vect{z}_1,\ldots,\vect{z}_r) \in  (\KK^{M})^{\times r}$, the
decomposition
\begin{equation}\label{eq:xrank_param_dec}
v = \xparam(\vect{z}_1) + \cdots + \xparam(\vect{z}_r) 
\end{equation}
is unique, i.e., the semialgebraic set 
\begin{equation}\label{eq:set_param_nonunique}
Y = \{ (\vect{z}_1,\ldots,\vect{z}_r) \in  (\KK^{M})^{\times r}  \,|\, \vect{v} \text{ in } \cref{eq:xrank_param_dec} \mbox{ has nonunique decompositions} \}.
\end{equation}
has Lebesgue measure zero.
\end{proposition}

Consider the case of Equation~\cref{eq:tens_decomposition_2}.
The Segre variety $Seg(\KK^{I_1}\times \KK^{I_2} \times \KK^{I_3})$ is $2$-identifiable if and only if the decomposition \cref{eq:tens_decomposition_2} is unique for general $\vect{a}_1,\vect{b}_1,\vect{c}_1,\vect{a}_2,\vect{b}_2,\vect{c}_2 $ (\emph{i.e.} drawn randomly with respect to an absolutely continuous probability distribution). Note the decomposition \cref{eq:tens_decomposition_2} is unique does not mean $\vect{a}_1, \dots, \vect{c}_2$ are unique, in fact they are unique up to scaling. Definition in the parameter space is more common in linear algebra and engineering literature. Hence \cref{cor:param_space} establishes correspondence between these two definitions. 

Finally, there is an important corollary of  \cref{def:rident_xrank} (in the case $\KK = \CC$) and \cref{cor:param_space} (in the case $\KK = \RR$).
\begin{corollary}
Let $\KK =\RR$ or $\CC$, $\widehat{X}$ satisfy assumptions of \cref{cor:param_space}.
If $\widehat{X}$ is $r$-identifiable, then
any vector $\vect{v} \in \Sigma_{r,\widehat{\set{X}}}$ is a limit of a sequence of vectors $\vect{v}_k \in \Sigma_{r,\widehat{\set{X}}}$ with a unique decomposition.
\end{corollary}
Thus, any rank-$r$ vector $\vect{v}$ can be approximated by rank-$r$  uniquely decomposable vectors to arbitrary precision. 
To our knowledge, in the case $\KK = \RR$, this fact is not explicitly mentioned in the literature.

\subsection{Necessary and sufficient conditions for generic uniqueness}
Here, in what follows, we consider only the case $\KK = \CC$. 
First, by \cite{Strassen83LAaiA-Rank}, 
if $\sigma_r(\widehat{X})$ is defective, then $\widehat{X}$ is not $r$-identifiable. If $\sigma_r(\widehat{X})$ is non-defective, then a general point in $\sigma_{r}(\widehat{X})$ has a finite number of decompositions. 
Thus, already looking at the dimension of $\sigma_{r}(\widehat{X})$ we can already conclude that $\widehat{X}$ is $r$-identifiable.
This can be done numerically using the Terracini's lemma.
\begin{lemma}[Terracini]\label{lem:terracini}
Assume that $\widehat{X}$ satisfies \crefrange{as:scale_invariant}{as:irreducible}.
Then for a general point $\vect{v} = \vect{p}_1, \dots, \vect{p}_r \in \sigma_r(\widehat{X})$, the tangent space is
\[
T_{\vect{v}} \sigma_r(\widehat{X})= \Span{T_{\vect{p}_1}\widehat{X}, \dots, T_{\vect{p}_r}\widehat{X}}.
\]
\end{lemma}

Hence, the non-defectivity can be checked numerically, by picking $r$ ``random'' points and comparing $\dim T_{\vect{v}} \sigma_r(\widehat{X})$ with $\exp\dim\sigma_r(\widehat{X})$.
A variety $\widehat{X}$ is called $r$-weakly defective if for $r$ general points in $\widehat{X}$ a general hyperplane tangent to them is tangent to $\var{X}$ elsewhere \cite{Chiantini.Ciliberto02-weaklydefective}.
If $\var{X}$ is not $r$-weakly defective, then $\var{X}$ is $r$-identifiable (the converse is not true).

\subsection{Examples: Veronese and Segre-Veronese varieties}
We review here some results on identifiability of varieties from \cref{tab:xrank_examples}, that will be needed.
First, recall a recent result that for all subgeneric ranks, the Veronese variety is $r$-identifiable.
\begin{theorem}[{\cite[Theorem 1.1]{Chiantini.etal16arxiv-generic}}]\label{thm:veronese_identifiability}
Let $d\ge 3$ and $m\ge 2$. Then $\nu_d(\CC^{m})$ is $r$-identifiable for all $r < r_2(m,d)$, where 
\begin{equation}\label{eq:veronese_identifiability}
r_2(m,d) = 
\begin{cases}
r_1(m,d)-1, & \text{if~} (m,d) \in \{(4,4), (3,6), (6,3) \}, \\
r_1(m,d), &  \mbox{otherwise}. \\
\end{cases}
\end{equation}
\end{theorem}

Next, we recall stronger results on $r$-weak defectivity of the Veronese varieties.
\begin{theorem}[{\cite{Ballico05CEJM-weak,Mella06-waring,Chiantini.Ciliberto06JLMS-Concept}}]\label{thm:veronese_wd}
Let $d\ge 3$ and $m \ge 2$. Then the Veronese variety $\nu_d (\CC^{m})$ is not $r$-weakly defective\footnote{The case $2\le m \le 3$ was proved in the proof of \cite[Thm 5.1]{Chiantini.Ciliberto06JLMS-Concept}, $d=3$ was proved in \cite[Thm. 4.1]{Mella06-waring}, the case $d\ge4$ is proved in \cite[Thm. 1.1.]{Ballico05CEJM-weak} (see also \cite[Corollary~4.5]{Mella06-waring}).} for $r < r_3(m,d)$, where 
\[
r_3(m,d) =
\begin{cases}
r_1(m,d)-\frac{m-2}{3}, &   d = 3, \\
r_1(m,d), &  \mbox{otherwise}. \\
\end{cases}
\]
\end{theorem}

For Segre-Veronese varieties, we are not aware of explicitly available results on identifiability.
However, the identifiability of such varieties can be easily deduced from \cref{thm:veronese_wd} and the results of \cite{Bocci.Chiantini.Ottaviani14-identifiability} on identifiability of Segre products of varieties.  Let 
\begin{equation}\label{eq:gen_rank_segre_veronese}
r_4(m,n,d)= \frac{\binom{m+d-1}{d}}{m+n-1}.
\end{equation}
\begin{corollary}\label{lem:segre_veronese_ident}
Let $m = \dim V \ge 2$, $d \ge 3$, $n = \dim W\ge 1$, and $k  < r_4(m,n,d)$,
where 
\begin{equation}\label{eq:X_identifiability_bound}
r_5(m,n,d) = 
\begin{cases}
r_2(m,d), & \text{if~} n=1, \\
\min\left(r_4(m,n,d), r_3(m,d)\right), & \text{if~} n > 1. \\
\end{cases}
\end{equation}
Then the variety $Seg(\nu_d (V) \times W)$  is $kn$-identifiable.
\end{corollary}
\begin{proof}
The proof is given in \cref{sec:basic_generic_equivalence}.
\end{proof}

Although the expression in \eqref{eq:X_identifiability_bound} looks complicated, in fact,
\[
r_5(m,n,d) = r_4(m,n,d)
\]
if $n > 1, d \ge 3$ or if $n=1, (m,d) \not\in \{(4,4), (3,6), (6,3) \}$.

\section{Veronese scrolls}\label{sec:algebraic}
In this section, we recall a variety that is a generalization of  the well-known rational normal scroll \cite{catalano1996-possible}.

\subsection{Simultaneous Waring decompositions}\label{sec:sim_waring}
Let $0\le a_1 \le \cdots \le a_d$ be a sequence of natural numbers\footnote{By convention, $\NN$ is the set of nonnegative integers and includes $0$.} put in one vector $\vect{a} = (a_1, \ldots,a_d) \in \NN^d$ 
and define a shorthand notation
\[
S^{\vect{a}} V \eqdef  S^{a_1} \set{V} \oplus S^{a_2} \set{V} \oplus \cdots \oplus S^{a_d}\set{V},
\]
which is a vector space of dimension
\[
\dim (S^{\vect{a}} V ) = \sum_{k=1}^{d} \binom{m+a_k-1}{a_k}.
\]
We say that $f = (f^{(1)},\ldots,f^{(d)}) \in S^{\vect{a}} V$ has a Waring-like decomposition of rank $r$ if there exist $\vect{v}_1, \ldots,\vect{v}_r$ and $c_{k,l}\in \KK$ such that
\begin{equation}\label{eq:sim_waring}
\begin{array}{rcl}
 f^{(1)} &=& c_{1,1} \vect{v}^{a_1}_1 + \cdots + c_{1,r} \vect{v}^{a_1}_r, \\
 &\vdots& \\
 f^{(d)} &=& c_{d,1} \vect{v}^{a_d}_1 + \cdots + c_{d,r} \vect{v}^{a_d}_r, 
\end{array}
\end{equation}
In other words, decomposition \eqref{eq:sim_waring} is equivalent to simultaneous Waring decompositions with the same vectors but different coefficients.

\begin{example}
Let us show that \cref{ex:n_1} is a special case of the Waring-like decomposition \eqref{eq:sim_waring}.
Since $f(\vect{0}) = \vect{0}$ in \eqref{eq:decoupling_one_output}, we have that 
\[
f(\vect{u}) = f^{(1)}(\vars) + \cdots + f^{(d)}(\vars),
\]
where $f^{(d)}(\vars)$ is the $d$-th degree homogeneous part of $f(\vars)$.
Hence, if the polynomial $f$  admits a decomposition  \eqref{eq:decoupling_one_output}, then all the homogeneous parts $f^{(d)}$ can be decomposed as
\[
\begin{array}{rcl}
f^{(1)} &=& c_{1,1} (\vect{v}_{1}^{\top} \vars) + \cdots + c_{1,r} (\vect{v}_{r}^{\top} \vars),\\ 
f^{(2)} &=& c_{2,1} (\vect{v}_{1}^{\top} \vars)^2 + \cdots + c_{2,r} (\vect{v}^\top_{r} \vars)^2,\\ 
&\vdots&\\
f^{(d)} &=& c_{d,1} (\vect{v}^{\top}_{1} \vars)^d + \cdots + c_{d,r} (\vect{v}^{\top}_{r}\vars)^d.\\ 
\end{array}
\]
which is a special case of \cref{eq:sim_waring} for the vector of integers $\vect{a} = (1,\ldots,d)$.
\end{example}

\subsection{Veronese scrolls: a parametric definition}
The decomposition \cref{eq:sim_waring} can be put in the framework of $X$-rank as follows.
Define the following map:
\begin{equation}\label{eq:polynomial_defining_map}
\begin{array}{rcl}
\psi:  \set{V} \times \KK^{d} & \to & S^{\vect{a}} V \\
(\vect{v},(c_1,\cdots,c_d)) & \mapsto & (c_1 \vect{v}^{a_1}, c_2 \vect{v}^{a_2},\ldots, c_d \vect{v}^{a_d}), 
\end{array}
\end{equation}
and define the image of this map as
\begin{equation}\label{eq:def_xhat}
\widehat{\set{X}}_{\vect{a}} = \widehat{\set{X}}_{\vect{a},V} \eqdef \psi(\set{V} \times \KK^{{d}}),
\end{equation}
and $\var{X}_{\vect{a}} = \var{X}_{\vect{a},\set{V}} \subset \PP S^{\vect{a}} V$ the corresponding subset in the projective space.

It is easy to see that $f=(f^{(1)},\ldots,f^{(d)}) \in S^{\vect{a}} V$ has a Waring-like decomposition if and only if it has  an  $X$-rank decomposition with $\widehat{X} = \widehat{X}_{\vect{a},V}$. It can be shown that $\widehat{X}_{\vect{a},V}$ satisfies \crefrange{as:scale_invariant}{as:irreducible} (affine cone of a  projective  variety $\widehat{X}_{\vect{a},V}$). In particular, when $\dim(V) = m = 2$, $\var{X}_{\vect{a},V}$ is the rational normal ($d$-fold) scroll, a classic object in algebraic geometry \cite{catalano1996-possible}. When $m>2$, we did not find a name of $\var{X}_{\vect{a},V}$ in the literatures, so we call it \emph{Veronese scroll}, as a hybrid of ``rational normal scroll'' and ``Veronese variety''. When $m \ge 2$, $\var{X}_{\vect{a},V}$ can be realized as a projective bundle\footnote{We are not reproducing the bundle construction, since it is difficult without going into technical details.} $\var{X}_{\vect{a},V} \simeq \mathbb{P} (\mathcal{O}_{\PP V} (a_1) \oplus \cdots \oplus \mathcal{O}_{\PP V} (a_d))$ \cite{Angelini.etal16arxiv-number,catalano1996-possible, Comon.etal15conf-polynomial}. In the following sections, we give explicit (ideal-theoretic) defining equations for the set \cref{eq:def_xhat}, which will provide an alternative proof that $\widehat{\var{X}}_{\vect{a},V}$ is a variety.

Now consider the following map
\[
\begin{array}{rcl}
\psi_m:  W \times V \times \KK^{d} & \to & S^{\vect{a}} V  \otimes W \\
(\mathbf{w},\mathbf{v},(c_1,\cdots,c_d)) & \mapsto & 
\begin{array}{rccccl}
(&w_1 c_1 \mathbf{v}^{a_1}, & b_1 c_2 \mathbf{v}^{a_1}, & \cdots & w_1c_d \mathbf{v}^{a_d}, &\\
&\vdots & & &\vdots &\\
&w_n c_1 \mathbf{v}^{a_1}, & b_m c_2 \mathbf{v}^{a_2}, & \cdots & w_nc_d \mathbf{v}^{a_d} &),\\
\end{array}
\end{array}
\]
and define $\widehat{Y}_{(a_1,\ldots,a_d)}$ the image of $\psi_m$.
It is easy to see that $\widehat{Y}_{(a_1,\ldots,a_d)} = Seg(\widehat{X}_{(a_1,\ldots,a_d)} \times W)$,
Moreover, as in \cref{sec:sim_waring}, we can show that the polynomial decomposition \cref{eq:decoupling_additive} is exactly the $X$-rank decomposition for $\widehat{Y}_{(1,\ldots,d)}$.

\subsection{Determinantal construction (defining equations)}\label{sec:determinantal}
This section is not needed to prove the main results of the paper, but still gives more insight in the nature of the Veronese scrolls.

First, recall a definition of the catalecticant matrix \cite[Ch.~1]{Iarobbino.Kanev99-Power}  (we prefer giving it in coordinates). 
Let $f\in S^dV$ be given by coordinates $\{f_{\boldsymbol{\alpha}}\}_{\boldsymbol{\alpha}\in \Delta_{s,m}}$, as defined in \cref{sec:symtdec}. Then the first catalecticant matrix, for $1 \le s \le d$, is defined as\footnote{In fact, this is the matrix representation map $S^{d-s} \set{V}^{*} \to S^{s} \set{V}^{*}$ given by differentiation.} 
\[
C_f \in \KK^{m \times \binom{m+d-2}{d-1}}, \text{ where} (C_f)_{i,\boldsymbol{\beta}} = f_{(\beta_1,\ldots,\beta_i+1,\ldots,\beta_m)},
\]
where the columns are indexed by $\boldsymbol{\beta} \in \Delta_{s,m-1}$.

\begin{proposition}\label{prop:determinantal}
Let $a_k \ge 1$, and ${f}=(f^{(1)},\ldots,f^{(d)}) \in S^{\vect{a}} V$.
Define the stacked  matrix as
\begin{equation}\label{eq:stacked_catalecticant}
S({f}) \eqdef \left[\begin{array}{c|c|c} C_{f_1} & \cdots & C_{f_d}  \end{array}\right].
\end{equation}
Then it holds that
\[
f \in \widehat{\set{X}}_{\vect{a},V} \iff \rank{S(f)} \le 1,
\]
\emph{i.e.} $\widehat{\set{X}}_{\vect{a},V}$ is defined (set-theoretically) by the vanishing of all
$2\times 2$ minors of $S({f})$.
\end{proposition}
\begin{proof}
The proof is contained in \cref{sec:defining_equations}.
\end{proof}

A similar construction for the matrix ${S(f)}$ can be found in \cite[\S 3]{Angelini.etal16arxiv-number}.

\begin{proposition}\label{prop:determinantal_ideal}
Let $a_k \ge 1$, and $S({f})$ be defined as in \cref{eq:stacked_catalecticant}.
Then the $2\times 2$ minors of $S({f})$ generate the ideal of $\widehat{\set{X}}_{\vect{a},V}$.
\end{proposition}

The proposition is much stronger than \cref{prop:determinantal}. 
The proof  relies on the tools of representation theory, and is contained in \cref{sec:defining_equations}.

\section{Main results}\label{sec:scrolls}
Throughout this section we assume that $\KK = \CC$. By \cite[Section~5]{Qi.etal16SJMAA-Semialgebraic}, all our results hold for the real case too.
We will also  use a shorthand $\var{X}_{\vect{a}}$ instead of $\var{X}_{\vect{a}, V}$.

\begin{remark}
A common idea to consider our model~(\ref{eq:decoupling_additive}) (suggested to us by one of the reviewers) is that decomposition \eqref{eq:decoupling_one_output} can be brought to the form \eqref{eq:waring_decomposition}, and hence  Waring decomposition can be applied (the same argument can be applied to bring \eqref{eq:decoupling_additive} to the form \eqref{eq:sim_waring_decomposition}).
However, homogenization can increase the number of terms, and does not give a good answer to our questions. 

For example, the homogenization of $f(x,y)$ in \cref{fig:polydec} is the trivariate polynomial
\[
6xy^2+4xy\xrightarrow{homogenization} 6xy^2+4xyz = xy(6y+z).
\]
But it is known \cite{Carlini.etal12JoA-solution} that this homogeneous polynomial does not have a Waring decomposition \cref{eq:waring_decomposition} with less that $4$ terms (compare with $3$ terms in \cref{fig:polydec}).
The reason for that is that the polynomials $g_1,g_2,g_3$ do not correspond to powers of linear forms for the homogenized polynomial.
In fact, homogenization restricts the form of polynomials $g_k$. We will study this model by investigating properties of Veronese scrolls.
\end{remark}

\subsection{Identifiability of Veronese scrolls and polynomial decompositions}
\begin{proposition}\label{prop:veronese_scroll_identifiability}
Let $m = \dim V \ge 2$, $a_d \ge 3$,  $n = \dim W \ge 1$.
Next, consider the Veronese scroll ${\widehat{X}}_{(a_1,\ldots,a_d)}$ with $\vect{a} = (a_1,\ldots,a_d)$, $1 \le a_1 \le \cdots \le a_d$, and the variety $\widehat{Y} = \widehat{Y}_{(a_1,\ldots,a_d)} $.
Then we have the following.
\begin{enumerate}
\item $\widehat{Y}$ is $r$-identifiable if
\begin{equation}\label{eq:identifiability_condition}
r\le\min(\ceil{r_5(m,n,a_d)}-1, \dim S^{a_1} V)n.
\end{equation}
\item  $\widehat{Y}$ cannot be $r$-identifiable for $r > n\dim(S^{a_1} V) $.
\end{enumerate}
\end{proposition}
 The proof is given in \cref{sec:proofs_identifiability}, and the idea of the proof is based on two facts:
\begin{enumerate}
\item Under the condition \eqref{eq:identifiability_condition}, the highest degree terms are generically unique, and $\vect{w}_k$ and $\vect{v}_k$ are uniquely determined.
\item The lower degree terms (coefficients $c_{k,l}$) can be recovered using a simple linear algebra.
\end{enumerate}

\cref{prop:veronese_scroll_identifiability} has immediate implications for the polynomial decomposition \eqref{eq:decoupling_additive}, which corresponds to the case where degrees are defined by $\vect{a} = (1,\ldots,d)$.

\begin{corollary}\label{prop:polydec_identifiability}
Let $d$, $m$, $n$ be such that $d \ge 3$, $m \ge 2$, and consider the field $\CC$ \begin{enumerate}
\item
The decomposition  \eqref{eq:decoupling_additive} is $r$-identifiable if
\begin{equation}\label{eq:theoretical_bound_identifiability}
r \le\min(m,\ceil{r_5(m,n,d)}-1)\cdot n.
\end{equation}
 In particular, if $m <r_5(m,n,d)$, then the model \eqref{eq:decoupling_additive} is $mn$-identifiable. 
\item 
The model \eqref{eq:decoupling_additive} cannot be $r$-identifiable for $r> mn$.
\end{enumerate}
\end{corollary}

First, let us give some examples. In Tables~\ref{tab:nbound3} and~\ref{tab:ndbound4}, we provide the calculated bound  on  identifiability \eqref{eq:theoretical_bound_identifiability} for $d=3,4$.
For comparison, we show the maximal non-defective rank obtained numerically\footnote{We also checked that the weak tangential nondefectivity described in  \cite{Chiantini.etal14SJMAA-Algorithm} holds for all cases Tables~\ref{tab:nd3} and~\ref{tab:nd4}, except when $r=mn$ (in that case, the weak tangential nondefectivity criterion works up to $mn-1$). } using  \cref{lem:terracini}.
In all tables, the cases when the rank coincides with $mn$ (i.e., the maximal possible  rank  by \cref{prop:polydec_identifiability}, part 2) are shown in bold.

\begin{table}[!hbt]
\centering
\subfloat[Our bound]{\label{tab:nbound3}
\pgfplotstableset{create on use/new/.style
  = {create col/expr accum={\pgfmathaccuma+1}{0}},
    my multistyler/.style={
    @my multistyler/.style={
    display columns/##1/.append style={column name={$##1$},
  postproc cell content/.code={%
  \pgfmathparse{int(equal((\pgfplotstablecol*(\pgfplotstablerow+1)),########1))}%
     \ifnum\pgfmathresult=1\relax%
     \pgfkeysalso{@cell content/.add={$\bf}{$}}\fi},}},
    @my multistyler/.list={#1}
  }
}
\pgfplotstableread{matlab/rank_non_def_bound3.txt}\mytable
\pgfplotstabletypeset[columns={new,0,1,2,3,4,5,6,7},skip rows between index={0}{1},skip rows between index={8}{12},every head row/.style={before row={\hline},after row=\hline},%
my multistyler={1,...,8}, 
every last column/.style={column type/.add={}{|}}, 
every last row/.style={after row=\hline},  
string replace={NaN}{},%
columns/new/.style={column name={\diagbox[width=2em, height=2em]{$m$}{$n$}},column type/.add={|}{|}}]\mytable}
\subfloat[Terraccini lemma]{\label{tab:nd3}
\pgfplotstableset{create on use/new/.style
  = {create col/expr accum={\pgfmathaccuma+1}{0}},
    my multistyler/.style={
    @my multistyler/.style={
    display columns/##1/.append style={column name={$##1$},
  postproc cell content/.code={%
  \pgfmathparse{int(equal((\pgfplotstablecol*(\pgfplotstablerow+1)),########1))}%
     \ifnum\pgfmathresult=1\relax%
     \pgfkeysalso{@cell content/.add={$\bf}{$}}\fi},}},
    @my multistyler/.list={#1}
  }
}  
\pgfplotstableread{matlab/rank_non_def3.txt}\mytable
\pgfplotstabletypeset[columns={new,0,1,2,3,4,5,6,7},skip rows between index={0}{1},skip rows between index={8}{12},every head row/.style={before row={\hline},after row=\hline},%
my multistyler={1,...,8}, 
every last column/.style={column type/.add={}{|}}, 
every last row/.style={after row=\hline},  
string replace={NaN}{},%
columns/new/.style={column name={\diagbox[width=2em, height=2em]{$m$}{$n$}},column type/.add={|}{|}}]\mytable}
\caption{Case $d=3$.}
\end{table}
\begin{table}[!hbt]
\centering
\subfloat[Our bound]{\label{tab:ndbound4}
\pgfplotstableset{create on use/new/.style
  = {create col/expr accum={\pgfmathaccuma+1}{0}},
    my multistyler/.style={
    @my multistyler/.style={
    display columns/##1/.append style={column name={$##1$},
  postproc cell content/.code={%
  \pgfmathparse{int(equal((\pgfplotstablecol*(\pgfplotstablerow+1)),########1))}%
     \ifnum\pgfmathresult=1\relax%
     \pgfkeysalso{@cell content/.add={$\bf}{$}}\fi},}},
    @my multistyler/.list={#1}
  }
}
\pgfplotstableread{matlab/rank_non_def_bound4.txt}\mytable
\pgfplotstabletypeset[columns={new,0,1,2,3,4,5,6,7},skip rows between index={0}{1},skip rows between index={8}{12},every head row/.style={before row={\hline},after row=\hline},%
my multistyler={1,...,8}, 
every last column/.style={column type/.add={}{|}}, 
every last row/.style={after row=\hline},  
string replace={NaN}{},%
columns/new/.style={column name={\diagbox[width=2em, height=2em]{$m$}{$n$}},column type/.add={|}{|}}]\mytable}
\subfloat[Terraccini's lemma]{\label{tab:nd4}
\pgfplotstableset{create on use/new/.style
  = {create col/expr accum={\pgfmathaccuma+1}{0}},
    my multistyler/.style={
    @my multistyler/.style={
    display columns/##1/.append style={column name={$##1$},
  postproc cell content/.code={%
  \pgfmathparse{int(equal((\pgfplotstablecol*(\pgfplotstablerow+1)),########1))}%
     \ifnum\pgfmathresult=1\relax%
     \pgfkeysalso{@cell content/.add={$\bf}{$}}\fi},}},
    @my multistyler/.list={#1}
  }
}  
\pgfplotstableread{matlab/rank_non_def4.txt}\mytable
\pgfplotstabletypeset[columns={new,0,1,2,3,4,5,6,7},skip rows between index={0}{1},skip rows between index={8}{12},every head row/.style={before row={\hline},after row=\hline},%
my multistyler={1,...,8}, 
every last column/.style={column type/.add={}{|}}, 
every last row/.style={after row=\hline},  
string replace={NaN}{},%
columns/new/.style={column name={\diagbox[width=2em, height=2em]{$m$}{$n$}},column type/.add={|}{|}}]\mytable}
\caption{Case $d=4$.}
\end{table}

As it is easy to see from \cref{tab:nbound3,tab:nd4},  that the bound given by \eqref{eq:theoretical_bound_identifiability} does not detect the maximum identifiability bound obtained by Terraccini's lemma (especially for $m < n$), but does perform well for the case $m \ge n$.
Moreover, the following remark can be made.

\begin{corollary}\label{cor:mn_ident_1}
$\phantom{x}$

\begin{enumerate}
\item 
For fixed $m$ and $n$, there exists $d_0$ such that the inequality $m < r_5(m,n,d)$  holds for all $d \ge d_0$.
\item 
If $m\ge n \ge 4$, then $m < r_5(m,n,d)$  holds true for all $d \ge 4$.
\end{enumerate}
\end{corollary}
\begin{proof}
\begin{enumerate}
\item This fact follows since the numerator of \eqref{eq:gen_rank_segre_veronese} for $m > 1$ is a strictly increasing in $d$.
\item If $5 \ge m \ge n \ge 4$, this can be verified from Table~\ref{tab:ndbound4}.
If $d\ge 4$ and $m\ge n > 1$, $m \ge 5$, then 
\[
\frac{r_5 (m,n,d)}{m} \ge \frac{r_5 (m,n,4)}{m} = \frac{(m+1)(m+2)(m+3)}{2\cdot 3 \cdot 4 \cdot (m+n-1)}  > 1.
\]
\end{enumerate}
\end{proof}

\begin{remark}\label{rem:DIS_generic_uniqueness}
The authors in \cite{Dreesen.etal14-Decoupling} suggest the bound
\begin{equation}\label{eq:uniqueness_bound_DIS}
(m-1)m(n-1)n \ge 2(r-1)r
\end{equation}
for decomposition \eqref{eq:decoupling_additive}, also shown in \cref{tab:ndbounddis}.
The bound \cref{eq:uniqueness_bound_DIS} appears from Kruskal-type generic uniqueness conditions for unstructured $m\times n\times N$ tensors \cite{DeLathauwer06SIMAX-Link}.
In fact, a better bound exists for unbalanced tensors, which is $(m-1)(n-1)$ \cite{Chiantini.etal14SJMAA-Algorithm}.

We make two remarks here:
\begin{enumerate}
\item The bound $mn$ is better than the heuristic bound \eqref{eq:uniqueness_bound_DIS} (see the values Table~\ref{tab:ndbounddis}).

\item The tensor considered in \cite{Dreesen.etal14-Decoupling} is structured, and 
the bound \eqref{eq:uniqueness_bound_DIS} cannot be directly applied to  model \eqref{eq:decoupling_additive} \footnote{Take for instance the simple case of symmetry. The maximal symmetric rank $R_s^o$ for which symmetric tensors will have a unique CP decomposition is smaller \cite{Chiantini.etal16arxiv-generic} than the maximal rank $R^o$ for which unconstrained tensors will have a unique CP decomposition \cite{ComoTDC09:laa,Chiantini.etal14SJMAA-Algorithm}.}.
In fact, for degree $2$ (see  \cref{tab:nd2}), the model can be non-identifiable even if the bound \eqref{eq:uniqueness_bound_DIS} holds.
\end{enumerate}
\end{remark}

\begin{table}[!hbt]
\centering
\subfloat[The heuristic bound given in \eqref{eq:uniqueness_bound_DIS}.]
{\label{tab:ndbounddis}
\pgfplotstableset{create on use/new/.style
  = {create col/expr accum={\pgfmathaccuma+1}{0}},
    my multistyler/.style={
    @my multistyler/.style={
    display columns/##1/.append style={column name={$##1$},
  postproc cell content/.code={%
  \pgfmathparse{int(equal((\pgfplotstablecol*(\pgfplotstablerow+1)),########1))}%
     \ifnum\pgfmathresult=1\relax%
     \pgfkeysalso{@cell content/.add={$\bf}{$}}\fi},}},
    @my multistyler/.list={#1}
  }
}
\pgfplotstableread{matlab/rank_non_def_bound_dis.txt}\mytable
\pgfplotstabletypeset[columns={new,0,1,2,3,4,5,6,7},skip rows between index={0}{1},skip rows between index={8}{12},every head row/.style={before row={\hline},after row=\hline},%
my multistyler={1,...,8}, 
every last column/.style={column type/.add={}{|}}, 
every last row/.style={after row=\hline},  
string replace={NaN}{},%
columns/new/.style={column name={\diagbox[width=2em, height=2em]{$m$}{$n$}},column type/.add={|}{|}}]\mytable}
\subfloat[Our bound]{\label{tab:nd2}
\pgfplotstableset{create on use/new/.style
  = {create col/expr accum={\pgfmathaccuma+1}{0}},
    my multistyler/.style={
    @my multistyler/.style={
    display columns/##1/.append style={column name={$##1$},
  postproc cell content/.code={%
  \pgfmathparse{int(equal((\pgfplotstablecol*(\pgfplotstablerow+1)),########1))}%
     \ifnum\pgfmathresult=1\relax%
     \pgfkeysalso{@cell content/.add={$\bf}{$}}\fi},}},
    @my multistyler/.list={#1}
  }
}  
\pgfplotstableread{matlab/rank_non_def2.txt}\mytable
\pgfplotstabletypeset[columns={new,0,1,2,3,4,5,6,7},skip rows between index={0}{1},skip rows between index={8}{12},every head row/.style={before row={\hline},after row=\hline},%
my multistyler={1,...,8}, 
every last column/.style={column type/.add={}{|}}, 
every last row/.style={after row=\hline},  
string replace={NaN}{},%
columns/new/.style={column name={\diagbox[width=2em, height=2em]{$m$}{$n$}},column type/.add={|}{|}}]\mytable}
\caption{Case $d=2$.}
\end{table}

In fact even if the model is non-identifiable, the decomposition can be partially unique.
\begin{corollary}\label{cor:polydec_identifiability_partial}
Let $s$ be a number $1 < s < d$ such that $\binom{m+s-1}{s}< r_5(m,n,d)$.
Then for all $r \le \binom{m+s-1}{s} n$, the decomposition \eqref{eq:decoupling_additive} is partially identifiable except the terms of degree less than $s$.
That is, all the elements in the decomposition \eqref{eq:decoupling_additive} can be determined uniquely (up to trivial indeterminacies), except the coefficients $c_{k,l}$, for $k < s$.
\end{corollary}
 
\subsection{Dimensions of secant varieties}

From Proposition~\ref{prop:polydec_identifiability} we can immediately find dimensions of secant varieties for small ranks.

\begin{proposition}\label{prop:veronese_scroll_dimensions}
Let $m = \dim V \ge 2$,  $n = \dim W \ge 1$, and  $\vect{a} = (a_1,\ldots,a_d)$, $1 \le a_1 \le \cdots \le a_{d-1} < a_d$, with $a_d \ge 3$.
Consider the variety $\widehat{Y} = \widehat{Y}_{(a_1,\ldots,a_d)}$, and assume that
\[
r \le (\ceil{r_5(m,n,a_d)}-1)n.
\]
Then we have that: 
\begin{enumerate}
\item If $r \le n \dim S^{a_1} V$, then $\widehat{Y}$ is non-defective, \emph{i.e.}
\[
\dim{\sigma}_r (\widehat{Y}) = \exp\dim{\sigma}_r (\widehat{Y}) = r(m+n+d-2),
\]
\item If $n  \dim S^{a_s} V < r \le n \dim S^{a_{s+1}} V$ then 
\begin{equation}\label{eq:veronese_scroll_defective_secant}
\sigma_r (\widehat{Y}) =  \left(S^{(a_1,\ldots,a_s)} V \otimes W \right)  \times {\sigma}_r (\widehat{Y}_{(a_{s+1}, \ldots, a_{d})}),
\end{equation}
and hence
\begin{align}
\dim{\sigma}_r (\widehat{Y}) & = \exp\dim{\sigma}_r (\widehat{Y}_{(a_{s+1}, \ldots, a_{d})}) + \dim \left(S^{(a_1,\ldots,a_s)} V \otimes W \right),\\
&=  r(m+n+d-2-s) + n \left(\sum\limits_{j=1}^{s}  \dim S^{a_j} V\right).
\end{align}
\end{enumerate}
\end{proposition}
The proof is based on \cref{prop:polydec_identifiability}, and is contained in \cref{sec:proofs_identifiability}.

It may be easier to look at the dimensions in terms of so-called defects of $\widehat{Y}$, defined as
\[
\delta_r(\var{Y}) \eqdef \exp\dim{\sigma}_r (\widehat{Y})- \dim{\sigma}_r (\widehat{Y}),
\]
where $\delta_r(\widehat{Y})$ is called the defect of ${\sigma}_r (\widehat{Y})$.
Then Proposition~\ref{prop:veronese_scroll_dimensions} can be reformulated as follows.
\begin{proposition}[Proposition~\ref{prop:veronese_scroll_dimensions} reformulated.]
Under the assumptions of Proposition~\ref{prop:veronese_scroll_dimensions}, the defect can the expressed as
\[
\delta_r(\widehat{Y}) = \sum\limits_{j=1}^d
\max(r-n\dim S^{a_j} V,0).
\]
\end{proposition}
\subsection{Generic ranks}\label{sec:generic_results}
In this section,  we consider only the case $n=1$, and $\vect{a} = (1,\ldots,d)$. From \cref{prop:veronese_scroll_dimensions} it follows that the behaviour of the ranks of secant varieties depends only on higher degrees.
As shown by the next lemma, for fixed $d$ and large $m$ everything depends on two higher degrees.
\begin{lemma}\label{lem:SdV_generic_maximal}
Let $d \ge 3$.
\begin{enumerate}
\item For all $m \ge 2$, it holds that $r_1(m,d) < \dim S^{d-1} V$.
\item For all $m > (d-2)(d-1)$ it holds that $\dim S^{d-2} V <  r_1(m,d)$.
\end{enumerate}
\end{lemma}
\begin{proof}$\phantom{x}$

\begin{enumerate}
\item First, for $d \ge 2$ and $m \ge 2$ it holds that $m+d-1 < md$.
Therefore,
\[
r_1(m,d) = 
  \binom{m+d-2}{d-1} \frac{m+d-1}{md} <  \binom{m+d-2}{d-1}.
\]

\item As in the previous item, we have that
\[
\frac{r_1(m,d)}{ \binom{m+d-3}{d-2}} = 
 \frac{(m+d-2) (m+d-1)}{m(d-1)d} >1
\]
The ratio is greater than one since $m \ge2$ and $m > (d-2)(d-1)$.
\end{enumerate}
\end{proof}

From \cref{prop:veronese_scroll_dimensions} and \cref{lem:SdV_generic_maximal}, we have the following immediate corollary.
\begin{corollary}
Under the conditions 1--2 in \cref{lem:SdV_generic_maximal},  we have
\[
r_{gen} (\widehat{X}_{(1,\ldots,d)})= r_{gen} (\widehat{X}_{(d-1,d)}).
\]
\end{corollary}
The main result in this subsection is on the bound on generic rank of $\widehat{X}_{(d-1,d)}$.

\begin{proposition}\label{prop:rgen_VS}
Let $d \ge 4$ and $m > 5$. Then
\[
\left\lceil
\frac{  \binom{m+d-2}{d-1} +   \binom{m+d-1}{d}}{m+1}\right\rceil \le r_{gen} (\widehat{X}_{(d-1,d)}) \le \left\lceil \frac{\binom{m+d-2}{d-1} + (m-1)\lceil \frac{\binom{m+d-1}{d}}{m} \rceil}{m} \right\rceil
\]
\end{proposition}
The lower bound just follows from \cref{cor:defectivity_r_gen}, the whole proof is given in \cref{sec:gen_ranks_proofs}.
For large $m$, the lower bound is exact.

\begin{proposition}\label{prop:rgen_VS_exact}
Let $d \ge 4$ and $m > (d-1)^2$. then
\[
r_{gen} = \left\lceil
\frac{  \binom{m+d-2}{d-1} +   \binom{m+d-1}{d}}{m+1}\right\rceil.
\]
\end{proposition}
In fact from the proof of Proposition~\ref{prop:rgen_VS_exact}, we can also obtain
\begin{proposition}\label{prop:Xd-1driden}
When $d \ge 4$, $m > (d-1)^2$, and $r < r_{gen}$, $\widehat{X}_{(d-1,d)}$ is $r$-identifiable.
\end{proposition}

As a corollary of \cref{prop:rgen_VS_exact} and \cref{thm:r_gen_r_max} we obtain the following bound on $r_{max}$ for polynomial decomposition \cref{eq:decoupling_one_output}.
\begin{corollary}\label{cor:X1_maximal}
Let $\KK = \CC$ or $\RR$, and fix $m$ and $d$ such that $d \ge 4$ and $m > (d-1)^2$. Then the maximal rank for the decomposition \cref{eq:decoupling_one_output} is bounded by
\begin{equation}\label{eq:bound_rmax_ours}
 r_{max} \le 2 \cdot \left\lceil
\frac{  \binom{m+d-2}{d-1} +   \binom{m+d-1}{d}}{m+1}\right\rceil.
\end{equation}
\end{corollary}
The bound in Corollary~\ref{cor:X1_maximal} implies that
\begin{equation}\label{eq:bound_rmax_ours_simplified}
r_{max} \le  
2 \left\lceil \binom{m+d-2}{d-1} \frac{m+2d-1}{(m+1)d}  \right\rceil \le \frac{2}{d}
\binom{m+d-2}{d-1} \left(1+ \frac{2(d-1)}{(m+1)}\right).
\end{equation}
Hence, the bound \eqref{eq:bound_rmax_ours_simplified} is  better than
\eqref{eq:bound_rmax_BBS}  if  $m>8$,  and  the ratio  between the bounds  \eqref{eq:bound_rmax_ours_simplified} and \eqref{eq:bound_rmax_BBS} approaches $\frac{2}{d}$ asymptotically as $m \to \infty$.

\section{Proofs}
\subsection{Basic results on generic uniqueness}\label{sec:basic_generic_equivalence}

\begin{proof}[Proof of \cref{lem:equiv_definition}]
Let
\[
\operatorname{Sec}_r^{\circ}(\widehat{X}) \eqdef \{(\vect{p}_1, \dots, \vect{p}_r, \vect{v})\in \widehat{X}^{\times r}\times V \colon \vect{p}_1, \dots, \vect{p}_r \in \widehat{X}, \vect{v} = \vect{p}_1 + \cdots + \vect{p}_r\},
\]
and let $\pi_1 \colon \widehat{X}^{\times r}\times V \to \widehat{X}^{\times r}$ and $\pi_2 \colon \widehat{X}^{\times r}\times V \to V$ be the projections. 
Observe that ${\sigma}_r(\widehat{X})$ is the Zariski closure of $\pi_2(\operatorname{Sec}_r^{\circ}(\widehat{X}))$, and $\pi_1 \colon \operatorname{Sec}_r^{\circ}(\widehat{X}) \to \widehat{X}^{\times r}$ is an isomorphism.

Then that $\sigma_r(\widehat{X})$ is identifiable implies $\pi_2 \colon \operatorname{Sec}_r^{\circ}(\widehat{X}) \to {\sigma}_r(\widehat{X})$ is birational, and thus the model $\widehat{X}$ is $r$-identifiable in the sense of \cref{eq:uniqueness_X}. On the other hand, if the model $\widehat{X}$ is $r$-identifiable in the sense of \cref{eq:uniqueness_X}, the cardinality of $\pi_2^{-1}\pi_2(p)$ is $1$ for a general $p \in \operatorname{Sec}_r^{\circ}(\widehat{X})$. Since $\pi_2(\operatorname{Sec}_r^{\circ}(\widehat{X}))$ contains a Zariski dense open subset of ${\sigma}_r(\widehat{X})$, then $\pi_2 \colon \operatorname{Sec}_r^{\circ}(\widehat{X}) \to {\sigma}_r(\widehat{X})$ is birational, which implies ${\sigma}_r(\widehat{X})$ is identifiable.
\end{proof}

\begin{proof}[Proof of \cref{cor:param_space}]
First, we consider the case $\KK = \CC$. 
By \cite[Exercise~II~3.22]{HartshorneAG77}, each $\vect{z}_i$ is general in $\CC^{M}$ if and only if $\xparam(\vect{z}_i)$ is general in $\widehat{X}$. Then the statement follows from Lemma~\ref{lem:equiv_definition}.

Next, we prove the statements in the case $\KK = \RR$, using  basic properties of semialgebraic sets. For convenience, we introduce the polynomial map $\mathscr{X}_r:  (\KK^{M})^{\times r} \to \ambspace$: 
\[
\mathscr{X}_r((\vect{z}_1,\ldots,\vect{z}_r)) = \xparam(\vect{z}_1) + \cdots + \xparam(\vect{z}_r).
\]
For the set $Y$ defined in \cref{eq:set_param_nonunique}, define $P = \mathscr{X}_r(Y)$.
The set $\Sigma_{r,\widehat{\set{X}}}$ is semialgebraic and denote its dimension by $d_1$.
The sets $Y$ and $P$ are also semialgebraic.

\fbox{``only if''} 
In this case, $\dim Y <  Mr$, and we need to prove that $\dim P < d_1$.
Suppose $\dim P = d_1$. Hence, there is an open ball $B \subset \ambspace$ such that $B \cap P = B \cap \Sigma_{r,\widehat{\set{X}}}$.
By continuity of the map $\mathscr{X}_r$, we have that $\mathscr{X}^{-1}_r(B) \subset Y$ is open, and hence $\dim \mathscr{X}^{-1}_r(B) = Mr$, hence a contradiction.

\fbox{``if''} 
In this case, we are given that $\dim P < d_1$ and we need to prove that $\dim Y < Mr$.
Suppose that it is not the case. 
By semialgebraic version of the Sard's theorem \cite[Lemma 2.1]{Qi.etal16SJMAA-Semialgebraic}, there exists an open ball $U \subset Y$ such that the rank of the Jacobian $J_{\mathscr{X}_r}$ is maximal (equal to $d_1$) on $U$. 
That implies that $\dim {\mathscr{X}_r (U)} = d_1$, which leads to a contradiction since ${\mathscr{X}_r (U)} \subset P$.
\end{proof}

In order to get results on identifiability of some Segre-Veronese varieties, we use a lemma that is a weaker version of the general result from \cite{Bocci.Chiantini.Ottaviani14-identifiability}.

\begin{lemma}[Corollary of {\cite[Lemma~3.1, Corollary~3.3]{Bocci.Chiantini.Ottaviani14-identifiability}.}]\label{lem:segre_identifiability}
Let $\var{X} \subset {\PP}A$  is a smooth non-degenerate projective variety, and $W$ be a vector space.
Let $\var{Y} = Seg(\var{X} \times {\PP}W)$ be the Segre embedding (such that $\dim({\var{Y}}) =  \dim({\var{X}}) + \dim(W) -1$).
If $\var{X}$ is not $r$-weakly defective, and
\begin{equation}\label{eq:condition_Segre_identifiability}
r(\dim(\var{Y})+1) < \dim(A),
\end{equation}
then $\var{Y}$ is $(r\cdot\dim(W))$-identifiable.
\end{lemma}

\begin{proof}[Proof of \cref{lem:segre_veronese_ident}]
For the case $n=1$, this is just Theorem~\ref{thm:veronese_identifiability}.
Now we consider $n>1$, and check the conditions of Lemma~\ref{lem:segre_identifiability}. In this case, the condition \eqref{eq:condition_Segre_identifiability} is equivalent to $r < r_4(m,n,d)$.
Since ${r_5(m,n,d)} = {\min (r_4(m,n,d), r_3(m,d))}$, the proof is complete.
\end{proof}

\subsection{Defining equations of Veronese scrolls}\label{sec:defining_equations}
\begin{proof}[Proof of \cref{prop:determinantal}]
$\boxed{\Rightarrow}$ This direction is evident. 
In this case $(f^{(1)},\ldots,f^{(d)}) = (c_1 \vect{v}^{a_1},\ldots,c_d \vect{v}^{a_d})$.
Since each $f_k$ is rank-one,
by \cite[Thm. 1.28]{Iarobbino.Kanev99-Power} each catalecticant matrix $C_{f^{(k)}}$ has rank $\le 1$.
Moreover the column space of each rank-one $C_{f^{(1)}}$ is spanned by the vector $\vect{v}$, therefore the column space of ${S(f)}$ is spanned by $\vect{v}$, and its rank does not exceed $1$.

$\boxed{\Leftarrow}$ Now consider ${S(f)}$ with rank $1$ (the case of rank $0$ is obvious). 
Define as $\vect{v}$ the vector that spans the column space of ${S(f)}$.
Since each of the matrices $C_{f^{(k)}}$ has rank $\le 1$,
from \cite[Thm. 1.28]{Iarobbino.Kanev99-Power} we have that $f = (f_1,\ldots,f_d) = (c_1 \vect{v}_1^{a_1},\ldots,c_d \vect{v}_d^{a_d})$.
But, from the apolarity \cite[Ch. 1]{Iarobbino.Kanev99-Power}, \cite{EhreR93:ejc}, all the vectors $\vect{v}_k$ must be collinear to $\vect{v}$.
Therefore $f \in \widehat{\set{X}}_{\vect{a},V}$.
\end{proof}

 Since $\widehat{X}_{\vect{a}, V}$ is invariant under the general linear group $\operatorname{GL}(V)$, each degree-$k$ component of the ideal $I(\sigma_r(\widehat{X}_{\vect{a}, V}))$, denoted by $I_k(\sigma_r(\widehat{X}_{\vect{a}, V}))$, in $S^k(S^{a_1}V \oplus \cdots \oplus S^{a_d}V)$ is a representation of $\operatorname{GL}(V)$. For any $V$,
\[
S^k(S^{a_1}V \oplus \cdots \oplus S^{a_d}V) = \bigoplus_{l_1 + \cdots + l_d = k} S^{l_1}(S^{a_1} V) \otimes \cdots \otimes S^{l_d}(S^{a_d} V),
\]
which is isomorphic to a direct sum of some irreducible representations $S_{\mu} V$ of $\operatorname{GL}(V)$, where $\mu$ is a partition of $l_1a_1 + \cdots + l_da_d$. Therefore, $I_k(\sigma_r(\widehat{X}_{\vect{a}, V}))$ is isomorphic to a direct sum of some $S_{\mu} V$'s. Let $S_{\overline{\mu}} V$ denote a special realization of $S_{\mu} V$ in $S^k(S^{a_1}V \oplus \cdots \oplus S^{a_d}V)$, see for example \cite[Section~6]{Landsberg12-Tensors} for more details. Similar to \cite[Proposition~4.4]{LandsbergManivel04FCM} we have

\begin{proposition}\label{prop:inheritance}
Given vector spaces $V, W$ with $r \le \dim V \le \dim W$ and $\dim V \ge 2$, then $S_{\overline{\pi}} V \subset I_k(\sigma_r(\widehat{X}_{\vect{a}, V}))$ if and only if $S_{\overline{\pi}} W \subset I_k(\sigma_r(\widehat{X}_{\vect{a}, W}))$.
\end{proposition}

\begin{proof} 
Given a basis $\{ \vect{v}_1, \dots, \vect{v}_{\dim V} \}$ for $V$, and a basis $\{ \vect{w}_1, \dots, \vect{w}_{\dim W} \}$ for $W$, fix an embedding $i \colon V \hookrightarrow W$ such that $i(\vect{v}_j) = \vect{w}_j$ for $1 \le j \le \dim V$. Since each irreducible representation is generated by its highest weight vector, then
\[
S_{\overline{\pi}} W = \operatorname{GL}(W) \cdot S_{\overline{\pi}} V
\]
for any $\pi$ with length $\ell(\pi) \le \dim V$ (See \cite{FultonHarris13-Rep, Landsberg12-Tensors}). The map $i$ induces an embedding
\[
\sigma_r(\widehat{X}_{\vect{a}, V}) \xhookrightarrow{i} \sigma_r(\widehat{X}_{\vect{a}, W}).
\]
So in $S^k(S^{a_1}W \oplus \cdots \oplus S^{a_d}W)$, we have $I_k(\sigma_r(\widehat{X}_{\vect{a}, W})) \subset I_k(i(\sigma_r(\widehat{X}_{\vect{a}, V})))$, which implies if $S_{\overline{\pi}} W \subset I_k(\sigma_r(\widehat{X}_{\vect{a}, W}))$ then $S_{\overline{\pi}} V \subset I_k(\sigma_r(\widehat{X}_{\vect{a}, V}))$.

Now we need to show for any $S_{\overline{\mu}} V \subset I_k(\sigma_r(\widehat{X}_{\vect{a}, V}))$, $S_{\overline{\mu}} W \subset I_k(\sigma_r(\widehat{X}_{\vect{a}, W}))$. Let
\begin{align}
\sigma_r^{\circ}(\widehat{X}_{\vect{a}, V}) \coloneqq \{& p \in \sigma_r(\widehat{X}_{\vect{a}, V}) \colon p = (c_{1,1} \vect{u}_1^{a_1}, \dots, c_{d,1} \vect{u}_1^{a_d}) + \cdots + (c_{1,r} \vect{u}_r^{a_1}, \dots, c_{d,r} \vect{u}_r^{a_d}), \\
& \text{where } \vect{u}_1, \dots, \vect{u}_r \text{ are linearly independent}\},
\end{align}
which is a Zariski dense open subset of $\sigma_r(\widehat{X}_{\vect{a}, V})$. Since $I(\sigma_r(\widehat{X}_{\vect{a}, V})) = I(\sigma_r^{\circ}(\widehat{X}_{\vect{a}, V}))$, we only need to show for any $f \in S_{\overline{\mu}} W \subset I_k(i(\sigma_r^{\circ}(\widehat{X}_{\vect{a}, V})))$, $f \in I_k(\sigma_r^{\circ}(\widehat{X}_{\vect{a}, W}))$. But this is true due to the fact $\operatorname{GL}(W) \cdot \sigma_r^{\circ}(\widehat{X}_{\vect{a}, V}) = \sigma_r^{\circ}(\widehat{X}_{\vect{a}, W})$. More precisely, for any $p \in \sigma_r^{\circ}(\widehat{X}_{\vect{a}, W})$, since there is some $g \in \operatorname{GL}(W)$ such that $g \cdot p \in i(\sigma_r^{\circ}(\widehat{X}_{\vect{a}, V}))$,
\[
f(p) = f(g^{-1} \cdot g \cdot p) = (g^{-1} \cdot f) (g \cdot p) = 0,
\]
which implies $f \in I_k(\sigma_r^{\circ}(\widehat{X}_{\vect{a}, W}))$.
\end{proof}

As a corollary of Proposition \ref{prop:inheritance} we have
\begin{proposition}\label{prop:idealveronesesscroll}
Given a vector space $V$ with $2 \le \dim V$, then
\[
S_{\overline{\pi}} V \subset I_k(\widehat{X}_{\vect{a}, V}) \iff S_{\overline{\pi}} V \subset I_k(\widehat{X}_{\vect{a}, \CC^2}).
\]
\end{proposition}
Since the ideal of the $\widehat{X}_{\vect{a}, \CC^2}$ is generated by $2 \times 2$ minors of $S(f)$ \cite[Proposition 4.5]{Eisenbud88AJM-linear},   we conclude that \cref{prop:determinantal_ideal}   is proved.

\subsection{Identifiability and dimensions of secant varieties of Veronese scrolls}\label{sec:proofs_identifiability}
\begin{proof}[Proof of \cref{prop:veronese_scroll_identifiability}]
$\phantom{x}$

\begin{enumerate}
\item
We have that 
\[
 S^{\vect{a}}V \otimes W \simeq (S^{a_1}V \otimes W) \oplus \cdots \oplus(S^{a_d}V \otimes W),
\]
and consider the $j$-th canonical projection $\pi_j: S^{\vect{a}}V \otimes W \to S^{a_j}V \otimes W$.
Let $r=kn$.

Consider $\var{\widehat{Y}}_{(a_d)} = \pi_{d}(\widehat{Y})$. 
Then, by properties of Zariski closures, we have that ${\sigma}_{r}(\widehat{Y}_{(a_d)}) = \overline{\pi_d({\sigma}_r(\widehat{Y}))}$,
and by \cref{cor:closure}, a general point in ${\sigma}_{r}(\var{\widehat{Y}}_{(a_d)})$ belongs to $\pi_d(\widehat{\sigma}_r(\var{Y}))$.

Hence, we can take a general element 
\[
\vect{p} = (f_1^{(1)},\ldots,f_n^{(1)}, \ldots, f_1^{(d)},\ldots,f_n^{(d)}) \in \widehat{\sigma}_{r}(\var{Y})
\]
such that ${\vect{p}} =  \vect{y}_1 + \cdots + \vect{y}_r$, $\vect{y}_k \in \widehat{\set{Y}}$ and the decomposition
\[
\pi_d({\vect{p}}) = \pi_d(\vect{y}_1) + \cdots + \pi_d(\vect{y}_r),
\]
is unique as $X$-rank decomposition with respect to $\var{\widehat{Y}}_{(a_d)}$ (due to $r$-identifiability of $\var{Z}$, which follows from \cref{lem:segre_veronese_ident}).
A general $\vect{y}_l \in \widehat{Y}$, has the form
\[
\vect{y}_l  =
(c_{1,l} \vect{v}^{a_1}_l \otimes \vect{w}_l, \cdots, c_{d-1,1} \vect{v}^{a_{d-1}}_l \otimes \vect{w}_l, \vect{v}^{a_d}_l \otimes \vect{w}_l).
\]
where  the vectors ($\vect{v}_l \otimes \vect{w}_l$) are determined uniquely, and
$\{\vect{v}_l\otimes \vect{w}_l\}_{l=1}^r$ are linearly independent since $r \le \dim (S^{a_1}V \otimes W)$.

Finally, the coefficients $c_{k,l}$ for $k < d$ should satisfy the equation
\begin{equation}\label{eq:lower_terms}
\pi_{k}(\vect{p}) = c_{k,1} \vect{v}^{a_k}_1 \otimes \vect{w}_1 + \cdots + c_{k,r} \vect{v}^{a_k}_r \otimes \vect{w}_r.
\end{equation}
By properties of Veronese embeddings, the vectors in  $\{\vect{v}_l^{a_k} \otimes \vect{w}_l\}_{l=1}^r$ are also linearly independent, and therefore $c_{k,l}$ are determined uniquely.  

\item Again, look at \eqref{eq:lower_terms} for $k=1$. We have that any system $\{\vect{v}_l^{a_1} \otimes \vect{w}_l\}_{l=1}^r$ is linearly dependent due to the fact that $r > \dim (S^{a_1}V \otimes W)$. Therefore, $\var{Y}_{(a_1,\ldots,a_d)}$ cannot be $r$-identifiable.
\end{enumerate}
\end{proof}

\begin{proof}[Proof of \cref{prop:veronese_scroll_dimensions}]
$\phantom{x}$

\begin{enumerate}
\item By Proposition~\ref{prop:veronese_scroll_identifiability},  ${\sigma}_r (\widehat{Y})$ is identifiable and thus nondefective.

\item It is sufficient to prove \cref{eq:veronese_scroll_defective_secant}, the rest follows automatically.
Let $\pi = \pi_{(s,\ldots,d)} :S^{(a_1,\ldots,a_d)} V \otimes W \to S^{(a_{s+1},\ldots,a_{d})} V \otimes W$ denote the canonical projection \emph{i.e.},
\[
\begin{array}{rcl}
\pi_{(s,\ldots,d)} \colon  (f_1^{(1)},\ldots,f_1^{(d)}, \ldots, f_n^{(1)},\ldots,f_n^{(d)})& \mapsto & (f_1^{(s)},\ldots,f_1^{(d)}, \ldots, f_n^{(s)},\ldots,f_n^{(d)}).
\end{array}
\]
As in the proof of Proposition~\ref{prop:veronese_scroll_identifiability},  we have that ${\sigma}_r(\widehat{Z}) = \overline{\pi({\sigma}_r(\widehat{Y}))}$ and by \cref{cor:closure} a general element ${\sigma}_r(\widehat{Z})$ can be taken from  $\pi({\sigma}_r(\widehat{Y}))$.

Next, as in Proposition~\ref{prop:veronese_scroll_identifiability}, there exists a Zariski-open subset of $\set{U} \subset \pi({\sigma}_r(\widehat{Y}))$ such that any $\widehat{u} \in \set{U}$ has the decomposition $\widehat{u} = \widehat{u}_1 + \cdots + \widehat{u}_r$, where 
\[
\widehat{u}_l  =
(c_{s+1,l} \vect{v}^{a_{s+1}}_l \otimes \vect{w}_l, \cdots, c_{d-1,l} \vect{v}^{a_{d-1}}_l \otimes \vect{w}_k, \vect{v}^{a_d}_l \otimes \vect{w}_l),
\]
and $v_l \otimes w_l$ are in general position. Therefore, for any $\widehat{p} \in S^{(a_1,\ldots,a_d)} V$ with $\pi(\widehat{p}) = \widehat{u}$ and  all $k \le s$, the equation
\eqref{eq:lower_terms} is always solvable.
Thus we have that
\[
\pi^{-1}(U) = (S^{(a_1,\ldots,a_s)} V \otimes W) \times U \subset {\sigma}_r(\widehat{Y}),
\]
and, moreover, ${\sigma}_r(\widehat{Z}) = \overline{\pi^{-1}(U)}  = \overline{(S^{(a_1,\ldots,a_s)} V \otimes W) \times U}$, which implies \eqref{eq:veronese_scroll_defective_secant}.
\end{enumerate}
\end{proof}

\subsection{Generic ranks}\label{sec:gen_ranks_proofs}
\begin{proof}[Proof of \cref{prop:rgen_VS}]
Recall the morphism $\Sigma_r$ defined by
\begin{align*}
\Sigma_r \colon (\widehat{{X}}_{(d-1,d)})^{\times r} &\to S^{d-1}V \oplus S^dV \\
\left((\mu_1 \vect{v}_1^{d-1}, \lambda_1 \vect{v}_1^d), \dots, (\mu_r \vect{v}_r^{d-1}, \lambda_r \vect{v}_r^d)\right) &\mapsto (\mu_1 \vect{v}_1^{d-1} + \cdots + \mu_r \vect{v}_r^{d-1}, \lambda_1 \vect{v}_1^d + \cdots + \lambda_r \vect{v}_r^d).
\end{align*}
Let $\pi_{d-1} \colon S^{d-1}V \oplus S^dV \to S^{d-1}V$ be the natural projection, and likewise for $\pi_d$. Then $r \ge r_{gen} (\var{X}^{(d-1,d)})$ if and only if
\[
\dim \pi_{d-1}(\pi_d^{-1}(p) \cap \operatorname{Im} \Sigma_r) = \dim S^{d-1}V
\]
for a general $p \in S^dV$. Since
\[
\dim \pi_{d-1}(\pi_d^{-1}(p) \cap \operatorname{Im} \Sigma_r) = \dim \operatorname{Im}(\Sigma_r) - \dim S^dV \le \dim (\widehat{{X}}_{(d-1,d)})^{\times r} - \dim S^dV,
\]
then $r \ge \frac{  \binom{m+d-2}{d-1} +   \binom{m+d-1}{d}}{m+1}$. On the other hand,
\[
\rank{p} = \rGenSdV \eqdef r_{gen}(\nu_{d}(\PP V)) = \lceil \frac{\binom{m+d-1}{d}}{m} \rceil,
\]
so we may assume $p = \vect{u}_1^d + \cdots + \vect{u}_{\rGenSdV}^d$ is a rank-$\rGenSdV$ decomposition of $p$. Then inside $\pi_d^{-1}(p) \cap \operatorname{Im} \Sigma_r$ there is a quasi-affine variety $Y$ parametrized by
\begin{align*}
Y &\eqdef \{ (\mu_1 \cdot \vect{u}_1^{d-1} + \cdots + \mu_{\rGenSdV} \cdot \vect{u}_{\rGenSdV}^{d-1} + \vect{u}_{\rGenSdV+1}^{d-1} + \cdots + \vect{u}_r^{d-1}, \\
&\qquad \vect{u}_1^{d} + \cdots + \vect{u}_{\rGenSdV}^{d} + 0 \cdot \vect{u}_{\rGenSdV+1}^{d} + \cdots + 0 \cdot \vect{u}_r^{d}) \in S^{d-1}V \oplus S^dV \colon \\
&\qquad   \mu_1,\dots,\mu_{\rGenSdV} \in \CC, \vect{u}_{\rGenSdV+1}, \dots, \vect{u}_r \in V \}.
\end{align*}
Since $\dim Y \le \dim \pi_d^{-1}(p) \cap \operatorname{Im} \Sigma_r$, $\dim \pi_{d-1}(Y) \le \dim \pi_{d-1}(\pi_d^{-1}(p) \cap \operatorname{Im} \Sigma_r) \le \dim S^{d-1}V$. Since $\rGenSdV \le \dim S^{d-1}V$, $p$ being general guarantees $\vect{u}_1^{d-1}, \dots, \vect{u}_{\rGenSdV}^{d-1}$ are linearly independent. Then when $r < \dim S^{d-1}V$ we can choose $\vect{u}_{\rGenSdV+1}, \dots, \vect{u}_r$ such that $\vect{u}_1^{d-1}, \dots, \vect{u}_r^{d-1}$ are linearly independent. By semicontinuity, for general $\vect{u}_{\rGenSdV+1}, \dots, \vect{u}_r$, we have $\vect{u}_1^{d-1}, \dots, \vect{u}_r^{d-1}$ are linearly independent. By Alexander-Hirschowitz theorem \cite{AlexanderHirschowitz95jag-nondefectivity}, when $r - \rGenSdV < r_{gen}(\nu_{d-1}(\PP V))$, the quasi-affine variety parametrized by
\[
\{\vect{u}_{\rGenSdV+1}^{d-1} + \cdots + \vect{u}_r^{d-1} \colon \vect{u}_{\rGenSdV+1}, \dots, \vect{u}_r \in V \},
\]
which contains an open Zariski subset of $\widehat{\sigma}_{r - \rGenSdV} (\nu_d(\PP V))$, has the expected dimension $(r - \rGenSdV)m$. Therefore 
\[
\pi_{d-1}(Y) = \{ \mu_1 \cdot \vect{u}_1^{d-1} + \cdots + \mu_{\rGenSdV} \cdot \vect{u}_{\rGenSdV}^{d-1} + \vect{u}_{\rGenSdV+1}^{d-1} + \cdots + \vect{u}_r^{d-1} \colon  \mu_1,\dots,\mu_{\rGenSdV} \in \CC,  \vect{u}_{\rGenSdV+1}, \dots, \vect{u}_r \in V \}
\]
has dimension $\rGenSdV + (r - \rGenSdV)m \le \dim S^{d-1}V$, which implies
\[r_{gen} (\var{X}_{(d-1,d)}) \le \bigg\lceil \frac{\binom{m+d-2}{d-1} + (m-1)\lceil \frac{\binom{m+d-1}{d}}{m} \rceil}{m} \bigg\rceil.\]
\end{proof}

\begin{proof}[Proof of \cref{prop:rgen_VS_exact}]
Consider the isomorphism:
\[
S^{d}(V\oplus \CC) \cong S^{(0,\ldots,d)} V.
\]
Then we have that $\nu_d(V\oplus \CC)$ is isomorphic to
\begin{equation}\label{eq:homogenized_variety}
\widehat{Z}_{d} = \widehat{Z}_{d,V} \eqdef \{(c^d,c^{d-1} \vect{v}, \ldots, c \vect{v}^{d-1}, \vect{v}^d) : c\in\CC, \vect{v} \in V \} \subset\widehat{X}_{(0,\cdots,d),V}.
\end{equation}
 Thus when $r > \dim S^{(0, \dots, d-2)}V$,   for any $\vect{p} \in S^{(0, \dots, d-2)}V$ and any general
\[
T = (c_1^d, \dots, \vect{v}_1^d) + \cdots + (c_r^d, \dots, \vect{v}_r^d) \in \sigma_r (\widehat{Z}_d),
\] 
there are some $\alpha_1, \dots, \alpha_r$ such that
\[
\vect{p} = \alpha_1 (c_1^d, \dots, c_1^{d-2} \vect{v}_1^{d-2}) + \cdots + \alpha_r (c_r^d, \dots, c_r^{d-2} \vect{v}_r^{d-2}) \in \sigma_r(\widehat{Z}_{d-2}) = S^{(0, \dots, d-2)}V,
\]
which implies
\[
\sigma_r(\widehat{Z}_d) = S^{(0, \dots, d-2)}V \oplus \sigma_r(\widehat{X}_{(d-1,d)}).
\]
By Alexander-Hirschowitz Theorem, when $r < \frac{\binom{m+d}{d}}{m+1}$, $\dim \sigma_r(\widehat{Z}_d) = r(m+1)$. Therefore
\[
\dim \sigma_r(\widehat{X}_{(d-1,d)}) = r(m+1) - \dim S^{(0, \dots, d-2)}V.
\]
 Since $m > (d-1)^2$, $ \left \lceil\frac{\binom{m+d-2}{d-1}+\binom{m+d-1}{d}}{m+1} \right \rceil \ge \dim S^{(0, \dots, d-2)}V$. 
In particular,
\[
r_{gen} (\widehat{X}_{(d-1,d)}) = \left \lceil\frac{\binom{m+d-2}{d-1}+\binom{m+d-1}{d}}{m+1} \right \rceil.
\]
\end{proof}

\begin{proof}[Proof of \cref{prop:Xd-1driden}]
When $r < \frac{\binom{m+d}{d}}{m+1}$, $\widehat{Z}_d$ is $r$-identifiable, which implies when $r < r_{gen}(\widehat{X}_{(d-1,d)})$, $\widehat{X}_{(d-1,d)}$ is $r$-identifiable.
\end{proof}

\section*{Acknowledgement} 
We would like to thank Ignat Domanov,  Philippe Dreesen, Mariya Ishteva, Giorgio Ottaviani and Nick Vannieuwenhoven  for enlightening discussions. We truly appreciate the help of the editors and the referees, their careful proofreading, and many thoughtful comments. 

\appendix
\section{Basic definitions}
\subsection{Symmetric tensors and homogeneous polynomials}\label{sec:symtdec}
Here we recall  some basic properties of symmetric tensors, which can be found in \cite{Comon.etal08SJMAA-Symmetric}.
A tensor $\mathcal{T} \in \KK^{m\times \cdots \times m}$ is called symmetric if 
\[
\mathcal{T}_{i_1,\ldots,i_d} =  \mathcal{T}_{i_{\pi(1)}, \ldots, i_{\pi(d)}}, 
\]
for any permutation of indices $\pi$. In this case, we write $\mathcal{T} \in S^d(\KK^m)$.
There is a one-to-one correspondence between symmetric tensors and homogeneous polynomials.
The contraction
\[
f(u) = \mathcal{T} \bullet_1 \vars \bullet_2 \vars \cdots \bullet_d \vars.
\]
gives a homogenous polynomial of degree $d$. 
Vice versa, any homogeneous polynomial corresponds to a unique element in $S^d(\KK^m)$ via polarization. 
In this paper, to avoid unnecessary extra symbols, for a homogeneous polynomial $f(\vars)$ we use the same letter for the corresponding $f\in S^d(\KK^m)$.

Next, a rank-one symmetric tensor of  order $d$ corresponds to the $d$-th power  of a linear form:
\[
(\vect{v} \otimes \cdots \otimes \vect{v}) \bullet_1 \vars \bullet_2 \vars \cdots \bullet_d \vars = (\vect{v}^{\top}\vars)^{d},
\]
which explains the equivalence between \eqref{eq:symtdec} and \eqref{eq:waring_decomposition}.

Finally, it is often convenient to give homogeneous polynomials in the following coordinates.
Let $\boldsymbol{\alpha} = (\alpha_1,\ldots,\alpha_m) \in \NN^m$ be a multi index\footnote{By convention, the set  $\NN$ includes $0$.}, we define the set
\[
\Delta_{s,m} \eqdef \{ \boldsymbol{\alpha} \in \NN^m : \alpha_1 + \cdots + \alpha_m = s\}.
\]
Now  the  homogeneous polynomial  $f \in S^d \set{V}$ can be represented in the following coordinates
\[
f(\vect{u}) = \sum_{\boldsymbol{\alpha} = (\alpha_1,\ldots,\alpha_m) \in \Delta_{d,m}} \frac{(\alpha_1+\cdots+\alpha_m)!}{\alpha_1!\cdots \alpha_m!} f_{\boldsymbol{\alpha}} \vect{u}^{\boldsymbol{\alpha}},
\]
where $\vect{u}^{\boldsymbol{\alpha}} = u_1^{\alpha_1} \cdots u_d^{\alpha_d}$.

\subsection{Algebraic varieties}\label{sec:varieties}
This subsection is devoted to a short summary of basic definitions in algebraic geometry that will be needed in this paper.
We choose a simplistic view, on a level of the popular book of Cox, Little and O'Shea \cite{Cox.etal97-Ideals}.
A quick and simple overview of the main definitions used here can be also found in the paper \cite{Sottile16-Real}.
As it was mentioned in the introduction, we only consider the case $\KK = \RR,\CC$. 

\begin{definition}[Algebraic variety]
A subset $\set{Z} \subseteq \KK^N$ is called an affine algebraic variety\footnote{As in \cite{Cox.etal97-Ideals}, we do not require a variety to be irreducible, contrary to some classic definitions.} if there exist a finite set of polynomials $p_1, \cdots, p_M \in \KK\brackets{\vect{z}}$  such that
\begin{equation}\label{eq:polynomial_system}
\vect{z} \in \set{Z} \iff \begin{cases}
p_1(\vect{z}) &= 0, \\
&\vdots\\
p_M(\vect{z}) &= 0 \\
\end{cases}
\end{equation}
\emph{i.e.} $\set{Z}$ is a zero locus of $p_1, \ldots, p_M$.
A set $\set{X}$ is a  called a proper subvariety of $Z$, if $X \subsetneq Z$ and $X$ is also a variety.
\end{definition}
\begin{remark}
 $\KK^N$ is also an algebraic variety: a zero locus of the zero polynomial $p(\vect{z}) \equiv 0$.
\end{remark}

\begin{definition}[Zariski closure]
For any set $Y \subset \KK^{N}$, by $\overline{Y}$ we denote the smallest algebraic variety $Z$, such that $Y \subseteq Z$. $\overline{Y}$ is called the Zariski closure of $Y$.
\end{definition}

\begin{definition}[Irreducibility]
A nonempty variety $X$  is called irreducible \cite{Cox.etal97-Ideals} it cannot be represented as a union of two distinct varieties.
(More precisely, if for a decomposition $X = Y \cup {Z}$ with $Y$, $Z$ varieties, it holds that either ${Y} \subseteq {Z} $ or ${Z} \subseteq {Y}$.) 
\end{definition}

\begin{definition}[Generic property]
We say that some property is \emph{generic} in an irreducible variety $Z$ if there exists a proper subvariety $\set{V} \subsetneq Z$ (of smaller dimension) such that the property is true  for all points in $Z \setminus V$.
\end{definition}

\begin{remark}[Generic properties in $\CC^{N}$]
If the property is generic in $\CC^{N}$, it implies\footnote{This follows from the fact that any proper algebraic subvariety has Lebesgue measure zero.} that a random vector in $\CC^N$ (drawn from  any absolutely continuous distribution)  satisfies a given generic property with probability $1$. 
\end{remark}

\begin{definition}
Let $Z$ be an irreducible variety in $\KK^{N}$, and $p_1, \ldots, p_M$ be a set of generators of its ideal.
Let $d$ be the maximal rank of the Jacobian matrix $J_{p}(\vect{z}) \eqdef [\frac{\partial{p}_{j}}{\partial z_j}]_{i,j=1}^{M,N}$ at $\vect{z}\in Z$.
Then the dimension is, by definition, $N - d$. 
A point $\vect{z}$ is called smooth if $J_{p}(\vect{z})$ has maximal rank at that point. 
Finally,  the dimension of a reducible variety is equal to the maximal dimension of its  irreducible components.
\end{definition}

\subsection{Polynomial images of algebraic varieties}
\begin{definition}
The set  $Z \subset \CC^{N}$  is called constructible, if it can be written as a finite union 
\[
Z = (X_1 \setminus Y_1) \cup \cdots \cup(X_r \setminus Y_r)
\]
where $X_k$, $Y_k$ are varieties.
\end{definition}
\begin{theorem}[Chevalley]
An image of a constructible set under a polynomial map  is constructible.
\end{theorem}

\begin{corollary}\label{cor:closure}
Assume that  $X\subset \CC^m$ is a variety,  $p: \CC^m \to \CC^n$  is a polynomial map and  $Y = \overline{p(X)}$, such that $Y$ is irreducible.
Then a general element in $Y$ lies in $p(X)$, \emph{i.e.} there exists a subvariety $Z \subsetneq Y$ of strictly smaller dimension such that $Y \setminus Z \in p(X)$.
\end{corollary}

\bibliographystyle{siam}
\bibliography{polydec-ranks}

\begin{thebibliography}{10}

\bibitem{AlexanderHirschowitz95jag-nondefectivity}
{\sc James Alexander and Andr\'e Hirschowitz}, {\em Polynomial interpolation in
  several variables}, Journal of Algebraic Geometry, 4 (1995), pp.~201--222.

\bibitem{AngeliniBocciChiantini16RICI}
{\sc E.~Angelini, C.~Bocci, and L.~Chiantini.}, {\em Real identifiability vs
  complex identifiability},  (2016).
\newblock Available from \url{https://arxiv.org/abs/1608.07197}.

\bibitem{Angelini.etal16arxiv-number}
{\sc E.~Angelini, F.~Galuppi, M.~Mella, and G.~Ottaviani}, {\em On the number
  of {Waring} decompositions for a generic polynomial vector},  (2016).
\newblock Available from \url{http://arxiv.org/abs/1601.01869}.

\bibitem{Ballico05CEJM-weak}
{\sc Edoardo Ballico}, {\em On the weak non-defectivity of veronese embeddings
  of projective spaces}, Central European Journal of Mathematics, 3 (2005),
  pp.~183--187.

\bibitem{Bernardi.etal16arxiv-real}
{\sc Alessandra Bernardi, Grigoriy Blekherman, and Giorgio Ottaviani}, {\em On
  real typical ranks}, tech. report, arxiv.org, 2015.
\newblock Available from \url{http://arxiv.org/abs/1601.01869}.

\bibitem{Bialynicki-Birula.Schinzel08CM-Representations}
{\sc Andrzej Bia{\l{}}ynicki-Birula and Andrzej Schinzel}, {\em Representations
  of multivariate polynomials as sums of polynomials in linear forms}, Colloq.
  Mathematicum, 112 (2008), pp.~201--233.

\bibitem{Blekherman.Teitler14-maximum}
{\sc Grigoriy Blekherman and Zach Teitler}, {\em On maximum, typical and
  generic ranks}, Mathematische Annalen, 362 (2015), pp.~1021--1031.

\bibitem{Bocci.Chiantini.Ottaviani14-identifiability}
{\sc Cristiano Bocci, Luca Chiantini, and Giorgio Ottaviani}, {\em Refined
  methods for the identifiability of tensors}, Annali di Matematica Pura ed
  Applicata (1923-), 193 (2014), pp.~1691--1702.

\bibitem{Carlini.etal12JoA-solution}
{\sc E.~Carlini, M.~V. Catalisano, and A.~V. Geramita}, {\em The solution to
  the {Waring} problem for monomials and the sum of coprime monomials}, Journal
  of Algebra, 370 (2012), pp.~5 -- 14.

\bibitem{catalano1996-possible}
{\sc Michael~L. Catalano-Johnson}, {\em The possible dimensions of the higher
  secant varieties}, American Journal of Mathematics,  (1996), pp.~355--361.

\bibitem{Chen.etal01SR-Atomic}
{\sc Scott~Shaobing Chen, David~L. Donoho, and Michael~A. Saunders}, {\em
  Atomic decomposition by basis pursuit}, SIAM Review, 43 (2001), pp.~129--159.

\bibitem{Chiantini.Ciliberto02-weaklydefective}
{\sc Luca Chiantini and Ciro Ciliberto}, {\em Weakly defective varieties},
  Transactions of the American Mathematical Society, 354 (2002), pp.~151--178.

\bibitem{Chiantini.Ciliberto06JLMS-Concept}
{\sc Luca Chiantini and Ciro Ciliberto}, {\em On the concept of k-secant order
  of a variety}, Journal of the London Mathematical Society, 73 (2006),
  pp.~436--454.

\bibitem{Chiantini.etal14SJMAA-Algorithm}
{\sc L.~Chiantini, G.~Ottaviani, and N.~Vannieuwenhoven}, {\em An algorithm for
  generic and low-rank specific identifiability of complex tensors}, SIAM
  Journal on Matrix Analysis and Applications, 35 (2014), pp.~1265--1287.

\bibitem{Chiantini.etal16arxiv-generic}
{\sc L.~Chiantini, G.~Ottaviani, and N.~Vannieuwenhoven}, {\em On generic
  identifiability of symmetric tensors of subgeneric rank}, Transactions of the
  American Mathematical Society,  (2016).
\newblock to appear.

\bibitem{ComoTDC09:laa}
{\sc Pierre Comon, Jos M. F.~Ten Berge, Lieven {De Lathauwer}, and Josephine
  Castaing}, {\em Generic and typical ranks of multi-way arrays}, Linear
  Algebra Appl., 430 (2009), pp.~2997--3007.

\bibitem{Comon.etal08SJMAA-Symmetric}
{\sc Pierre Comon, Gene~H. Golub, Lek-Heng Lim, and Bernard Mourrain}, {\em
  Symmetric tensors and symmetric tensor rank}, SIAM. J. Matrix Anal. Appl., 30
  (2008), pp.~1254--1279.

\bibitem{Comon.etal15conf-polynomial}
{\sc P.~Comon, Y.~Qi, and K.~Usevich}, {\em A polynomial formulation for joint
  decomposition of symmetric tensors of different orders}, in Latent Variable
  Analysis and Signal Separation, E.~Vincent, A.~Yeredor, Z.~Koldovsk\'{y}, and
  P.~Tichavsk\'{y}, eds., vol.~9237 of Lecture Notes in Computer Science,
  Springer, 2015, pp.~22--30.

\bibitem{Cox.etal97-Ideals}
{\sc David Cox, John Little, and Donald O'Shea}, {\em Ideals, Varieties and
  Algorithms: An Introduction to Computational Algebraic Geometry and
  Commutative Algebra}, Springer, 2nd~ed., 1997.

\bibitem{DeLathauwer06SIMAX-Link}
{\sc L.~{De Lathauwer}}, {\em A link between the canonical decomposition in
  multilinear algebra and simultaneous matrix diagonalization}, SIAM Journal on
  Matrix Analysis and Applications, 28 (2006), pp.~642--666.

\bibitem{Dreesen.etal14-Decoupling}
{\sc Philippe Dreesen, Mariya Ishteva, and Johan Schoukens}, {\em Decoupling
  multivariate polynomials using first-order information}, SIAM. J. Matrix
  Anal. Appl., 36 (2015), pp.~864--879.

\bibitem{EhreR93:ejc}
{\sc Richard Ehrenborg and Gian-Carlo Rota}, {\em Apolarity and canonical forms
  for homogeneous polynomials}, European Jour. Combinatorics, 14 (1993),
  pp.~157--181.

\bibitem{Eisenbud88AJM-linear}
{\sc David Eisenbud}, {\em Linear sections of determinantal varieties},
  American Journal of Mathematics, 110 (1988), pp.~541--575.

\bibitem{FultonHarris13-Rep}
{\sc William Fulton and Joe Harris}, {\em Representation theory: a first
  course}, Springer Science \& Business Media, 2013.

\bibitem{Giri.Bai10-Block}
{\sc F.~Giri and E.W. Bai}, {\em Block-oriented Nonlinear System
  Identification}, Lecture Notes in Control and Information Sciences, Springer,
  2010.

\bibitem{HartshorneAG77}
{\sc Robin Hartshorne}, {\em Algebraic geometry}, Springer-Verlag, New
  York-Heidelberg, 1977.
\newblock Graduate Texts in Mathematics, No. 52.

\bibitem{Iarobbino.Kanev99-Power}
{\sc Anthony Iarrobino and Vassil Kanev}, {\em Power sums, Gorenstein Algebras
  and Determinantal Loci}, vol.~1721 of Lecture Notes in Mathematics, Springer,
  1999.

\bibitem{Landsberg12-Tensors}
{\sc J.~M. Landsberg}, {\em Tensors: Geometry and applications}, vol.~128,
  American Mathematical Soc., 2012.

\bibitem{LandsbergManivel04FCM}
{\sc Joseph~M. Landsberg and Laurent Manivel}, {\em On the ideals of secant
  varieties of {S}egre varieties}, Foundations of Computational Mathematics, 4
  (2004), pp.~397--422.

\bibitem{Logan.Shepp75D-Optimal}
{\sc Benjamin~F. Logan and Larry~A. Shepp}, {\em Optimal reconstruction of a
  function from its projections}, Duke Math. J., 42 (1975), pp.~645--659.

\bibitem{Mella06-waring}
{\sc Massimiliano Mella}, {\em Singularities of linear systems and the waring
  problem}, Transactions of the American Mathematical Society, 358 (2006),
  pp.~5523--5538.

\bibitem{Oskolkov02SMJ-Representations}
{\sc Konstantin~I. Oskolkov}, {\em On representations of algebraic polynomials
  as a sum of plane waves}, Serdica Mathematical Journal,  (2002),
  pp.~379--390.

\bibitem{Qi.etal16SJMAA-Semialgebraic}
{\sc Yang Qi, Pierre Comon, and Lek-Heng Lim}, {\em Semialgebraic geometry of
  nonnegative tensor rank}, SIAM Journal on Matrix Analysis and Applications,
  37 (2016), pp.~1556--1580.

\bibitem{Schinzel02JdTdNdB-decomposition}
{\sc Andrzej Schinzel}, {\em On a decomposition of polynomials in several
  variables}, Journal de Th\'{e}orie de Nombres de Bordeaux, 14 (2002),
  pp.~647--666.

\bibitem{Schinzel02CM-decomposition}
\leavevmode\vrule height 2pt depth -1.6pt width 23pt, {\em On a decomposition
  of polynomials in several variables, ii}, Colloquium Mathematicum, 92 (2002),
  pp.~67--79.

\bibitem{Schoukens.etal14conf-System}
{\sc Johan Schoukens, Anna Marconato, Rik Pintelon, et~al.}, {\em System
  identification in a real world}, in IEEE 13th International Workshop on
  Advanced Motion Control (AMC), March 2014, pp.~1--9.

\bibitem{Shin.Ghosh96IToNN-Ridge}
{\sc Yoan Shin and Joydeep Ghosh}, {\em Ridge polynomial networks}, IEEE
  Transactions on Neural Networks, 6 (1995), pp.~610--622.

\bibitem{Sottile16-Real}
{\sc Frank Sottile}, {\em Real algebraic geometry for geometric constraints},
  tech. report, 2016.
\newblock arXiv preprint 1606.03127.

\bibitem{Strassen83LAaiA-Rank}
{\sc V.~Strassen}, {\em Rank and optimal computation of generic tensors},
  Linear Algebra and its Applications, 5253 (1983), pp.~645 -- 685.

\bibitem{Usevich14conf-Decomposing}
{\sc Konstantin Usevich}, {\em Decomposing multivariate polynomials with
  structured low-rank matrix completion}, in 21st Int. Symposium on
  Mathematical Theory of Networks and Systems, July 7-11, 2014. Groningen, The
  Netherlands, 2014, pp.~1826--1833.

\bibitem{VanMulders.etal14conf-Identification}
{\sc Anne {Van Mulders}, Laurent Vanbeylen, and Konstantin Usevich}, {\em
  Identification of a block-structured model with several sources of
  nonlinearity}, in Proceedings of the 14th European Control Conference (ECC
  2014), 2014, pp.~1717--1722.

\bibitem{Zak2004}
{\sc Fyodor~L. Zak}, {\em Determinants of projective varieties and their
  degrees}, in Algebraic Transformation Groups and Algebraic Varieties, V.~L.
  Popov, ed., Springer, Berlin, 2004, pp.~207--238.

\end{thebibliography}

\end{document}